\documentclass[onecolumn,accepted=2020-07-29,a4paper]{quantumarticle}
\pdfoutput=1

\usepackage[numbers,sort&compress]{natbib}
\usepackage{graphicx}
\usepackage{dcolumn}
\usepackage{dsfont}
\usepackage{enumitem}
\usepackage{bm}
\usepackage{amsmath,amssymb,amsthm}
\usepackage{braket}
\usepackage{cancel}

\usepackage{algpseudocode,algorithm,algorithmicx}

\usepackage{xcolor}
\usepackage{url}
\usepackage{microtype}
\usepackage{soul}

\usepackage{changepage}

\renewcommand{\vec}[1]{\ensuremath{\boldsymbol{#1}}}
\newcommand{\vt}{\vec{\theta}}
\newcommand{\ml}{\mathcal{L}}
\newcommand{\nl}{\lvert\lvert}
\newcommand{\nr}{\rvert\rvert}
\newcommand{\hc}{\operatorname{h.\!c.}}

\newtheorem{theorem}{Theorem}
\newtheorem{example}{Example}

\newtheorem{lemma}[theorem]{Lemma}

\newtheorem{observation}{Observation}

\newtheorem{definition}{Definition}

\begin{document}

\title{Stochastic gradient descent for hybrid quantum-classical optimization}

\author{Ryan Sweke}
\thanks{These authors contributed equally to this work.}
\affiliation{\mbox{Dahlem Center for Complex Quantum Systems, Freie Universit\"{a}t Berlin, 14195 Berlin, Germany}}
\orcid{0000-0002-6202-8864}

\author{Frederik Wilde}
\thanks{These authors contributed equally to this work.}
\affiliation{\mbox{Dahlem Center for Complex Quantum Systems, Freie Universit\"{a}t Berlin, 14195 Berlin, Germany}}
\orcid{0000-0002-6224-1964}

\author{Johannes Jakob Meyer}
\affiliation{\mbox{Dahlem Center for Complex Quantum Systems, Freie Universit\"{a}t Berlin, 14195 Berlin, Germany}}
\orcid{0000-0003-1533-8015}

\author{Maria Schuld}
\affiliation{\mbox{Xanadu, 777  Bay  Street,  Toronto,  Ontario,  Canada}}
\affiliation{\mbox{Quantum Research Group, University of KwaZulu-Natal, 4000 Durban, South Africa}}
\orcid{0000-0001-8626-168X}
\
\author{\\Paul K. Faehrmann}
\affiliation{\mbox{Dahlem Center for Complex Quantum Systems, Freie Universit\"{a}t Berlin, 14195 Berlin, Germany}}
\orcid{0000-0002-8706-1732}

\author{Barth\'{e}l\'{e}my Meynard-Piganeau}
\affiliation{\mbox{Department of Physics, Ecole Polytechnique, Palaiseau, France}}

\author{Jens Eisert}
\affiliation{\mbox{Dahlem Center for Complex Quantum Systems, Freie Universit\"{a}t Berlin, 14195 Berlin, Germany}}
\affiliation{\mbox{Helmholtz Center Berlin, 14109 Berlin, Germany}}
\affiliation{\mbox{Department of Mathematics and Computer Science, Freie Universit{\"a}t Berlin, D-14195 Berlin}}
\orcid{0000-0003-3033-1292}


\begin{abstract}
Within the context of hybrid quantum-classical optimization, gradient descent based optimizers typically require the evaluation of expectation values with respect to the outcome of parameterized quantum circuits. In this work, we explore the consequences of the prior observation that estimation of these quantities on quantum hardware results in a form of \emph{stochastic} gradient descent optimization. We formalize this notion, which allows us to show that in many relevant cases, including VQE, QAOA and certain quantum classifiers, estimating expectation values with $k$ measurement outcomes results in optimization algorithms whose convergence properties can be rigorously well understood, for any value of $k$. In fact, even using single measurement outcomes for the estimation of expectation values is sufficient. Moreover, in many settings the required gradients can be expressed as linear combinations of expectation values -- originating, e.g., from a sum over local terms of a Hamiltonian, a parameter shift rule, or a sum over data-set instances -- and we show that in these cases $k$-shot expectation value estimation can be combined with sampling over terms of the linear combination, to obtain ``doubly stochastic'' gradient descent optimizers. For all algorithms we prove convergence guarantees, providing a framework for the derivation of rigorous optimization results in the context of near-term quantum devices. Additionally, we explore numerically these methods on benchmark VQE, QAOA and quantum-enhanced machine learning tasks and show that treating the stochastic settings as hyper-parameters allows for state-of-the-art results with significantly fewer circuit executions and measurements.  
\end{abstract}

\maketitle

\section{Introduction}

\noindent Hybrid quantum-classical optimization with parameterized quantum circuits \cite{Hybrids} provides a promising approach for understanding and exploiting the potential of 
\emph{noisy intermediate-scale quantum (NISQ)} devices \cite{Preskill2018quantumcomputingin}. In this approach, which includes large classes
of well studied methods like \textit{variational 
quantum eigensolvers} (VQE) \cite{peruzzo2014variational}, the \textit{quantum approximate optimization algorithm} (QAOA) 
\cite{farhi2014quantum} and \textit{quantum classifiers} \cite{SchuldClassification, FarhiClassification}, a classical 
optimization scheme is utilized to update 
the parameters of a hybrid quantum-classical 
model, and developing and understanding optimization techniques tailored 
to this setting is of natural importance \cite{benedetti2019parameterized}. 
While a variety of gradient-free optimization methods have been proposed and studied \cite{zhu2018training}, in this work we will be concerned with gradient descent type optimizers. This focus is motivated by the fact that for highly over-parameterized classical models (such as modern neural networks) gradient-based optimization offers many advantages over gradient-free methods \cite{Goodfellow-et-al-2016}, and as a result developing and analyzing such methods designed specifically for the quantum-classical setting is crucial, particularly as the computational power of available near-term quantum devices increases.

In this hybrid quantum-classical setting, since part of the hybrid computation is executed by a quantum circuit, the optimized loss is generically a function of the expectation values of quantum observables. While zeroth-order methods can be immediately applied to obtain an approximation to the gradient from evaluations of the loss function \cite{kandala2017hardware,harrow_napp}, there also now exist a variety of strategies to directly evaluate the gradient, either via distinct quantum algorithms \cite{gilyen2019optimizing,verdon2018universal}, or through the measurement of suitable observables with respect to states generated by the parameterized model \cite{bergholm2018pennylane, mitarai2018quantum,schuld2019evaluating}. As an example of the latter, it has recently been shown that the partial derivatives required for gradient descent can be exactly expressed as linear combinations of the same expectation values appearing in the loss function, but with respect to states generated from a shift in the tunable parameters of the parameterized quantum circuit -- a strategy called the `parameter shift rule' \cite{mitarai2018quantum, schuld2019evaluating}. However, as an infinite number of measurements is required for the exact evaluation of a particular expectation value, it is not possible to implement \emph{exact} gradient descent in this hybrid quantum-classical setting, even when the exact gradient can be written as a function of expectation values. As a result, previous approaches have typically used large numbers of measurements to estimate expectation values as accurately as possible \cite{Kandala,leyton2019robust,havlivcek2019supervised}, and the necessary circuit repetitions represent a significant overhead for the implementation of gradient based optimizers.

Recently however it has been observed that using a finite number of measurements for the evaluation of gradients effectively results in the implementation of \textit{stochastic} gradient descent \cite{harrow_napp}. Stochastic gradient descent (SGD)
optimization works by replacing the exact partial derivative at each optimization step with an \emph{estimator} of the partial derivative, and when the estimator is unbiased it is often possible to prove rigorous convergence guarantees in appropriately simplified settings \cite{Shalev,karimi2016linear}. Additionally, SGD is the method of choice for the vast majority of large-scale machine learning models, where it has been found to offer advantages over exact gradient descent, such as faster evaluation of the gradient, faster convergence, and avoidance of local minima \cite{bottou2010large, bottou2008tradeoffs, kleinberg2018alternative,ruder2016overview}. Given the natural connection between SGD and hybrid quantum-classical optimization, Harrow and Napp \cite{harrow_napp} provide explicit algorithms for the construction of unbiased estimators for the gradient of typical loss functions arising in the context of hybrid quantum-classical optimization, both by exploiting single shot measurements for the estimation of expectation values, and by importance sampling single terms of linear combinations of expectation values, which arise naturally from the use of local Hamiltonians to construct loss functions. Additionally, using these estimators they are able to generalize a variety of existing upper and lower SGD convergence bounds into the hybrid quantum-classical setting. These bounds then allow them to construct a simple class of optimization problems for which it can be proven that certain first-order hybrid quantum-classical SGD methods converge substantially faster than any zeroth-order method.

 In light of these previous results and observations, we explore in this work, both theoretically and numerically, the applicability of these techniques, and a variety of heuristic extensions, in multiple concrete settings. As previously observed by Harrow and Napp \cite{harrow_napp}, one such setting is in the context of VQE, where the cost function is typically a linear combination of expectation values of local Hamiltonian terms, and as such an unbiased estimator of the gradient can be constructed by sampling commuting subsets of these local terms in each optimization step. However, multiple other hybrid quantum-classical optimization settings also facilitate the use of similar strategies. For example, the parameter shift rule  re-expresses partial derivatives as linear combinations of expectation values with respect to shifted circuit parameters, and unbiased estimators to these derivatives can therefore similarly be easily obtained by sampling subsets of terms in each optimization step. This insight can be particularly valuable in the context of continuous variable quantum circuits, where an infinite number of parameter shift terms may be required \cite{schuld2019evaluating}. Additionally, in many data driven settings, such as quantum classifiers for example, loss functions are constructed as sums over data-set instances, and sampling subsets of instances in each optimization step is a well known classical strategy -- known as mini-batch SGD -- which has already been adopted from classical machine learning \cite{SchuldClassification}. Of particular interest is the fact that these ``sampling from linear combination" type SGD schemes can all be combined with each other, and with efficient expectation value estimation, resulting in optimization methods that we refer to as \emph{``doubly stochastic"} gradient descent. This is in the same spirit as similarly motivated techniques for kernel methods, in which scalable doubly stochastic estimators to the gradient are constructed through the combination of two distinct unbiased approximations to the gradient \cite{dai2014scalable,ds2}.

Similarly to the kernel method setting, the concrete doubly stochastic SGD optimization schemes we formulate here may provide large efficiency gains over existing approaches, crucial for the implementation of variational algorithms on existing devices. For example, if a quantum classifier trained with $D$ data points has a cost function consisting of $M$ observables, which each need $K$ ``parameter shift'' terms to extract the analytical gradient, and each expectation is an average of $N$ measurement results, the most extreme version of stochastic gradient descent uses only \textit{a single measurement sample} to estimate the gradient, saving $\mathcal{O}(DMKN)$ measurements in each optimization step. Additionally, from a theoretical perspective, as has been previously mentioned, formalizing these notions allows one to prove convergence guarantees, in suitably restricted settings, for both existing methods and newly proposed SGD optimizers, thereby placing gradient-based hybrid quantum-classical optimization within a rigorous theoretical framework, which can be exploited both to guide the development and analysis of new methods, and to facilitate comparison of existing methods~\cite{harrow_napp}.

As suggested by previously obtained convergence bounds, and confirmed here in numerical simulations of various benchmark tasks for which these convergence bounds may not be directly applicable, the increased variance of such extreme-case estimators typically results in more optimization steps being required for convergence, and potentially non-optimal final solutions. However, while more optimization steps may be required, the enhanced efficiency of each each optimization step can result in significant overall savings in the number of circuit executions and measurements which are necessary to achieve convergence. Furthermore, in practice, a promising strategy is to treat the various SGD parameters -- such as learning rate, number of measurement shots, and number of linear combination terms sampled -- as hyper-parameters which are adjusted appropriately through the course of optimization. As we see from numerical simulations, basic implementations of this strategy suggest one is able to converge to highly accurate solutions, while retaining the efficiency gains of SGD. In fact, very recently the authors of Ref.~\cite{kubler2019adaptive} have also observed that the number of measurement shots used for expectation value estimation can be treated as a hyper-parameter within an SGD framework, and suggested multiple sophisticated heuristic strategies, inspired by state-of-the-art methods from classical optimization, for varying this parameter through the course of optimization. These strategies may well also be applicable within the context of the doubly stochastic gradient descent schemes we consider here, and provide a natural avenue of investigation for future work.

This work is structured as follows: We present the setting and basic idea in Section \ref{s:setting}. We then develop a more general framework and body of concrete examples by exploring in detail the settings of VQE, QAOA and quantum classifiers in Sections \ref{s:vqe}, \ref{s:qaoa} and \ref{s:mse} respectively. In Section \ref{s:extensions}, we discuss certain extensions beyond the settings we consider here, before proceeding in Section~\ref{s:convergence} to prove rigorous convergence guarantees, complimentary to previously obtained bounds, for all considered optimizers. Given this formal framework, we present in Section \ref{s:results} multiple numerical experiments and benchmarks, before concluding in Section \ref{s:conclusion} with both a discussion and outlook.

\section{Setting and Idea}\label{s:setting}

\noindent Given a model parameterized by $\vt \in \mathbb{R}^d$, and some loss function 
$\ml: \mathbb{R}^d\rightarrow \mathbb{R}$, stochastic gradient descent algorithms can be viewed as optimization algorithms in which the exact gradient descent update rule

\begin{equation}
    \vt^{(t+1)} = \vt^{(t)} - \alpha \nabla \ml(\vt^{(t)}),
\end{equation}
is replaced with a stochastic update rule of the form
\begin{equation}\label{e:sgd_update}
    \vt^{(t+1)} = \vt^{(t)} - \alpha g^{(t)}(\vt^{(t)}),
\end{equation}
where $\{g^{(t)}(\vt)\}$ is a sequence of random variables -- estimators of the gradient -- which defines the particular algorithm. Perhaps counter-intuitively, this may offer multiple advantages: Stochasticity can potentially aid in the avoidance of local minima and saddle points \cite{bottou1991stochastic}, and if well designed, $g^{(t)}(\vt)$ can be much more efficiently evaluated than the full gradient $\nabla\ml$ \cite{zinkevich2010parallelized,recht2011hogwild}. While such algorithms are heavily used and studied as heuristics, there is also an increasing body of work seeking to understand and prove their convergence properties. While we postpone a detailed discussion of convergence theorems to Section \ref{s:convergence}, we note that a fundamental property required for obtaining convergence guarantees is that the estimators $\{g^{(t)}(\vt)\}$ are \emph{unbiased} -- i.e.,  
\begin{equation}
\mathbb{E}[g^{(t)}(\vt)] = \nabla \ml(\vt)
\end{equation}
for all $t$. In this work we will be concerned with the development of stochastic gradient descent algorithms within the setting of hybrid quantum-classical optimization. In particular, we will consider loss functions ${\cal L}$
of the form
\begin{equation}
    \ml(\vt) = \ml(\vt, \langle O_1 \rangle_{\vt}, \ldots, \langle O_M \rangle_{\vt} ),
\end{equation}
where $\langle O_i \rangle_{\vt}$ is the expectation value of an observable $O_i$ with respect to the outcome of a parameterized quantum circuit $U(\vt)$ acting on an initial state vector $|\vec{0}\rangle$, i.e.
\begin{equation}
    \langle O_i \rangle_{\vt} = \langle \vec{0}|U^{\dagger}(\vt)O_iU(\vt)|\vec{0}\rangle.
\end{equation}
In the following sections, we will examine in detail specific loss functions relevant for current applications, however in order to present some of the basic ideas, let us consider as a first example the simple loss function 
\begin{equation}
    \mathcal{L}(\vt) = \langle O\rangle_{\vt},
\end{equation}
for some observable $O$, which we assume can be readily measured. In order to utilize gradient descent optimization it is necessary to obtain expressions for all partial derivatives $\partial\mathcal{L}(\vt)/\partial \theta_i$. While zeroth-order methods, such as finite differences, could be used to estimate these partial derivatives from evaluations of the loss function, we will in this work focus on first-order methods, in which one calculates these partial derivatives directly, without necessarily evaluating the loss function as an intermediate step. In particular, for many settings of interest, it has recently been shown that a \emph{parameter shift rule} can be derived, via which all partial derivatives can be expressed as linear combinations of the same expectation value, but with respect to slightly shifted circuit parameters \cite{mitarai2018quantum, schuld2019evaluating}. In this work we primarily restrict ourselves to settings in which such a rule can be derived\footnote{All results still hold if the parameter shift rule has to be replaced with simple numerical differentiation. However, noise inherent to the quantum device may be much more detrimental in this case.}, which we formalize via the following definition:

\begin{definition}[Parameter shift rule]\label{def:pshift}
A quantum circuit $U(\vt)$ parameterized by $\vt \in \mathbb{R}^d$, satisfies a $K$-term parameter shift rule if for all observables $O$ and for all parameters $\theta_i$, with $i \in [1,\ldots, d]$, there exist some $\{\gamma_{k,i}\}$ and $\{c_{k,i}\}$ such that
\begin{equation}\label{e:parameter_shift}
    \frac{\partial }{\partial \theta_i}\langle O\rangle_{\vt} =\sum_{k = 1}^K \gamma_{k,i} \langle O\rangle_{\vt_{k,i}},
\end{equation}
where $\vt_{k,i} =  \vt + c_{k,i}\vec{e}_i $, with $\vec{e}_i$ denoting a unit vector in the $i$'th direction.
\end{definition}
\noindent In order to facilitate intuition, before continuing let us consider the following example:
\begin{example}[Parameter shift rule for single qubit generators \cite{schuld2019evaluating}]\label{example:ps} Consider a parameterized quantum circuit of the form $U(\vt) = \prod_{i = 1}^d e^{-i\theta_iG_i}$, where each $G_i$ is a single qubit Hermitian operator with eigenvalues $\pm r_i$. In this case one finds that
\begin{equation}
    \frac{\partial }{\partial \theta_i}\langle O\rangle_{\vt} = r_i \langle O\rangle_{\vt + \left(\frac{\pi}{4r_i}\right)\vec{e}_i} -  r_i\langle O\rangle_{\vt - \left(\frac{\pi}{4r_i}\right)\vec{e}_i}
\end{equation}
\end{example}

\noindent Now, under the assumption of a $K$-term parameter shift rule for circuit ansatz $U(\vt)$, we see that all partial derivatives $\partial\mathcal{L}(\vt)/\partial \theta_i$ of our example loss function are the linear combination of $K$ expectation values, as per Eq.~\eqref{e:parameter_shift}. As expectation values cannot be evaluated exactly on quantum devices, we see already at this stage that exact gradient descent is not possible in this setting. However, as we show now, estimating expectation values via a finite number of measurement outcomes leads to unbiased estimators for the gradient, and therefore to well motivated stochastic gradient descent schemes. In particular, let us start with the following general definition for the $n$-shot sample mean estimator of an expectation value:

\begin{definition}[$n$-sample mean estimator]\label{d:sample_mean}
Given a parameterized quantum circuit $U(\vt)$ (with $\vt \in \mathbb{R}^d$) and an observable $O$ we define $\tilde{o}^{(n)}(\vt)$ as the $n$-sample mean estimator of $\langle O \rangle_{\vt}$ -- i.e., the estimator of $\langle O \rangle_{\vt}$ obtained by averaging the results of $n$ measurements of the observable $O$ on the state vector $U(\vt)|0\rangle$.
\end{definition}
\noindent Note that by construction we have that
\begin{equation}
\mathbb{E}[\tilde{o}^{(n)}(\vt)] = \langle O \rangle_{\vt},
\end{equation}
for all $n$, and that while all $n$-sample mean estimators have the same expectation value the variance of the estimator decreases with increasing $n$. Given this unbiased estimator for a single expectation value, it now follows straightforwardly that such estimators can be linearly combined to obtain unbiased estimators for the required partial derivatives, i.e.,
\begin{equation}
\mathbb{E}[\sum_{k = 1}^K\gamma_{k,i}\tilde{o}^{(n)}(\vt_{k,i})] = \sum_{k = 1}^K\gamma_{k,i}\mathbb{E}[\tilde{o}^{(n)}(\vt_{k,i})]  = \sum_{k = 1}^K\gamma_{k,i}\langle O\rangle_{\vt_{k,i}} = \frac{\partial }{\partial \theta_i}\langle O\rangle_{\vt},
\end{equation}
and therefore the estimator $g_i(\vt) = \sum_{k}\gamma_{k,i}\tilde{o}^{(n)}(\vt_{k,i})$ is an unbiased estimator for $\partial\mathcal{L}(\vt)/\partial \theta_i$, which requires $nK$ measurements to construct. Using this estimator we can then define a well founded SGD optimization algorithm via the update rule
\begin{equation}
    \theta_i^{(t+1)} = \theta_i^{(t)} - \alpha g_i(\vt^{(t)}),
\end{equation}
where we have implicitly defined $g^{(t)}_i(\vt^{(t)}) := g_i(\vt^{(t)})$ for all $t$ - i.e. the same estimator is used for all update steps. Importantly, note that this algorithm is valid even in the extreme case of $n=1$. In fact, as exact evaluation of expectation values is not possible, many previous gradient-based approaches to loss functions such as the one considered here have been large $n$ instances of such an SGD algorithm \cite{Kandala,leyton2019robust,havlivcek2019supervised}. It is, 
however, possible to go further, and define a ``doubly stochastic" gradient descent optimizer by not only estimating the expectation values via $n$ measurements, but also sampling subsets of terms from the linear combination in each optimization step, and applying appropriate correction weights. In the extreme case of sampling only single terms, such an optimizer requires only $n$ measurements per optimization step, as opposed to $nK$ measurements. 

In order to formalize this notion, we will denote the set of non-negative integers less than or equal to $k$ as $[k] := \{1,\ldots,k\}$, and let us assume for now that the estimator $\tilde{o}^{(n)}(\vt_{k,i})$ for each term $\langle O\rangle_{\vt_{k,i}}$ in the linear combination of Eq.~\eqref{e:parameter_shift} is a discrete random variable which can take $n(k,i)$ values $\{\lambda^{(k,i)}_j\, |\, j \in [n_{k,i}]\}$, with respective probabilities $\mathrm{prob}(\lambda^{(k,i)}_j)$. We now note that the random variable $g_i(\vt)$ taking values $\{(K\gamma_{k,i})\lambda^{(k,i)}_j\, |\,  j \in [n(k,i)], k\in [K]\}$ with probabilities $(1/K)\mathrm{prob}(\lambda^{(k,i)}_j)$ is an unbiased estimator for the linear combination of Eq.~\eqref{e:parameter_shift}, i.e.,
\begin{align}
    \mathbb{E}[g_i(\vt)] &= \sum_{k = 1}^K\sum_{j = 1}^{n(k,i)} [(1/K)\mathrm{prob}(\lambda^{(k,i)}_j)][(K\gamma_{k,i})\lambda^{(k,i)}_j]\\
    &= \sum_{k = 1}^K\gamma_{k,i}\sum_{j = 1}^{n(k,i)} \mathrm{prob}(\lambda^{(k,i)}_j)\lambda^{(k,i)}_j\\
    &= \sum_{k = 1}^K\gamma_{k,i}\mathbb{E}[\tilde{o}^{(n)}(\vt_{k,i})]\\
    &= \sum_{k = 1}^K\gamma_{k,i}\langle O\rangle_{\vt_{k,i}}\\
    &= \frac{\partial }{\partial \theta_i}\langle O\rangle_{\vt}.
\end{align}
The important thing to note is that the random variable $g_i(\vt)$ can be sampled by first sampling a term of the linear combination uniformly at random -- i.e., drawing $k$ from $[1,\ldots,K]$ with $\mathrm{prob}(k) = 1/K$ -- and then sampling the random variable $\tilde{o}^{(n)}(\vt_{k,i})$ (requiring $n$ measurements) before applying the correction factor $K\gamma_{k,i}$ to the outcome. Alternatively, one can also sample with $\mathrm{prob}(k) = |\gamma_{k,i}|/(\sum_k|\gamma_{k,i}|)$, and apply the correction factor $\gamma_{k,i}/\mathrm{prob}(k)$. While for certain settings this may well lead to estimators with lower variance, for ease of presentation we restrict ourselves here to uniform sampling over linear combinations. As $g_i(\vt)$ is an unbiased estimator for $\partial\mathcal{L}(\vt)/\partial \theta_i$ one can use this estimator for an SGD optimizer, requiring only $n$ measurements per optimization step, which is summarized via the following algorithms:

\begin{algorithm}[H]
  \caption{Stochastic gradient descent
    \label{alg:stochastic_gradient_base}}
  \begin{algorithmic}[1]
    \State Set initial circuit parameters $\vt^{(0)} \in \mathbb{R}^d$ and learning rate $\alpha^{(0)}$
    \State $t\gets 0$
    \While{$t < T $} \Comment{Iterate through optimization steps}
        \ForAll{ $1 \leq i \leq d$} \Comment{Iterate through circuit parameters}
        	\State $g_i(\vt^{(t)}) \gets \Call{PartialDerivativeEstimator}{i,\vt^{(t)}}$ \Comment{Construct estimator for $\partial\mathcal{L}(\vt^{(t)})/\partial \theta_i$}
        \EndFor
    \State $\alpha^{(t+1)} \gets \Call{GetLearningRate}{t,\alpha^{(t)}}$ \Comment{Update the learning rate}
    \State $\theta_i^{(t+1)} \gets \theta_i^{(t)} - \alpha^{(t+1)} g_i(\vt^{(t)})$ \Comment{Update all circuit parameters} \label{l:update_rule}
    \State $t \gets t+1$
    \EndWhile
  \end{algorithmic}
\end{algorithm}

\begin{algorithm}[H]
  \caption{$n$-shot doubly stochastic partial derivative estimator
    \label{alg:nshot_dsgd_estimator}}
  \begin{algorithmic}[1]
  	\Statex Given $U(\vt)$ satisfying a $K$-term parameter shift rule, and $\mathcal{L}(\vt) = \langle O\rangle_{\vt}$
    \Statex
    \Function{PartialDerivativeEstimator}{$i,\vt^{(t)}$}
    \State Sample $k \in \{1,\ldots, K\}$ with $\mathrm{prob}(k) = 1/K$ \Comment{Sample a parameter shift term}
        	\State Evaluate $\tilde{o}^{(n)}(\vt^{(t)}_{k,i})$  \Comment{ Construct estimator for $\langle O\rangle_{\vt^{(t)}_{k,i}}$ via $n$ measurements}
        	\State $g_i(\vt^{(t)}) \gets (K\gamma_{k,i})\tilde{o}^{(n)}(\vt^{(t)}_{k,i})$ \Comment{ Construct estimator for $\partial\mathcal{L}(\vt^{(t)})/\partial \theta_i$ by applying correction factor}
    \State \Return $g_i(\vt^{(t)})$
    \EndFunction
  \end{algorithmic}
\end{algorithm}

\noindent At this stage, the basic idea should be clear: In order to implement stochastic gradient descent one requires estimators for all partial derivatives, and when these partial derivatives can be expressed via linear combinations of expectation values then one can define a simple unbiased estimator for the gradient by using a finite number of measurements to estimate the expectation values occurring in the linear combination. Moreover, one can define a ``doubly stochastic" estimator, as illustrated in Algorithm \ref{alg:nshot_dsgd_estimator}, by additionally sampling a subset of the terms of the linear combination in each optimization step. However, as we will see in the following sections, such linear combinations can arise also from loss functions that are a linear combination of observables, as in VQE, or from loss functions which are a sum over data-set instances, as for example found in quantum classification problems. In the latter example, one has to exercise caution as the loss function may not be a simple linear combination of expectation values, but rather a linear combination of non-linear functions of expectation values. Generic settings such as these require care to deal with properly, as discussed in Sections \ref{s:mse} and \ref{s:extensions}. Finally, before continuing we note that we have not specified the function \textproc{GetLearningRate} which is called by Algorithm \ref{alg:stochastic_gradient_base}. While a constant learning rate could be used, In practice a variety of different adaptive strategies are exploited, which can be proven to improve convergence properties in restricted settings \cite{karimi2016linear,pmlr-v89-li19c}. Additionally, it should be noted that a wide variety of more sophisticated variations to Algorithm \ref{alg:stochastic_gradient_base} exist, in which both the learning rate and parameter update rule depend explicitly on the history of prior gradient estimates \cite{ruder2016overview}. We will explore and compare such variations to Algorithm \ref{alg:stochastic_gradient_base} in Section \ref{s:results}.

\section{Unbiased estimators for VQE}\label{s:vqe}

In this section, we examine how the ideas introduced in Section \ref{s:setting} can be straightforwardly extended and applied to the setting of variational quantum eigensolvers (VQE). In particular, given a local Hamiltonian 
\begin{equation}
H = \sum_{j = 1}^MH_j,
\end{equation}
which encodes a problem of interest, we define the VQE loss function
\begin{equation}\label{e:vqe_loss}
    \ml(\vt) =  \langle H \rangle_{\vt} = \sum_{j = 1}^M\langle H_j \rangle_{\vt}.
\end{equation}
For this simple loss function, we have that
\begin{equation}\label{e:grad_vqe}
    \frac{\partial \ml(\vt)}{\partial \theta_i} = \frac{\partial \langle H \rangle_{\vt}}{\partial \theta_i} = \sum_{j = 1}^M \frac{\partial \langle H_j \rangle_{\vt}}{\partial \theta_i}
\end{equation}
and it follows that unbiased estimators for $\partial\langle H_j \rangle_{\vt}/\partial \theta_i$ can be straightforwardly linearly combined to construct an unbiased estimator for $\partial \ml(\vt)/\partial \theta_i$.  From the previous section it should however be clear that, provided $U(\vt)$ satisfies a parameter shift rule, one can easily construct either singly or doubly stochastic estimators for all $\partial\langle H_j \rangle_{\vt}/\partial \theta_i$ via Eq. \eqref{e:parameter_shift}, replacing $O$ with $H_j$. Once again, these estimators for $\partial\langle H_j \rangle_{\vt}/\partial \theta_i$ can then either be directly combined to provide an unbiased estimator for $\partial \ml(\vt)/\partial \theta_i$, or one can further sample over linear combinations of subsets of local Hamiltonian terms. To be clear, let us assume that $U(\vt)$ satisfies a $K$-term parameter shift rule, i.e., for all $H_j$
\begin{equation}
    \frac{\partial\langle H_j \rangle_{\vt}}{\partial \theta_i} = \sum_{k = 1}^K \gamma_{k,i} \langle H_j \rangle_{\vt_{k,i}}.
\end{equation}
We then have that
\begin{equation}
    \frac{\partial \mathcal{L}(\vt)}{\partial \theta_i} = \sum_{j = 1}^M\sum_{k = 1}^K \gamma_{k,i} \langle H_j \rangle_{\vt_{k,i}}.
\end{equation}
As a result, the simplest unbiased estimator for $\partial \ml(\vt)/\partial \theta_i$ that we can define is just the linear combination of estimators for $\langle H_j \rangle_{\vt_{k,i}}$ -- i.e., the random variable 
\begin{equation}
    g_i(\vt) = \sum_{j = 1}^M\sum_{k = 1}^K \gamma_{k,i} \tilde{h}_j^{(n)}(\vt_{k,i}),
\end{equation}
where $\tilde{h}_j^{(n)}(\vt_{k,i})$ is the $n$-sample mean estimator of $\langle H_j \rangle_{\vt_{k,i}}$, as per Definition \ref{d:sample_mean}.
This estimator, which requires $nKM$ measurements per parameter update, can be constructed using the function~ \textproc{PartialDerivativeEstimator\_1} given in Algorithm \ref{alg:VQE_estimators} below.

However, as we have seen in the previous section, in order to obtain more efficient unbiased estimators for $\partial \ml(\vt)/\partial \theta_i$, we can sample subsets of terms from either the summation over parameter shift terms, the summation over local Hamiltonian terms, or both. In order to provide an explicit example, consider the function \textproc{PartialDerivativeEstimator\_2} in Algorithm \ref{alg:VQE_estimators}, in which only a single local term of the Hamiltonian is sampled, with a cost of $nK$ measurements per parameter update. Note the similarity of the approach taken here with \emph{random compiling} in digital quantum simulation \cite{Earl}, which in that context offers advantages such as favourable scaling of Trotter error accumulation. While this function explicitly samples a single local term of the Hamiltonian, it is possible to straightforwardly generalize this to an algorithm which samples a \emph{subset} of local terms of the Hamiltonian for each parameter update. In particular, note that if $[H_l,H_m] = 0$, then the estimators $\tilde{h}_l^{(n)}(\vt_{k,i})$ and $\tilde{h}_m^{(n)}(\vt_{k,i})$ could both be constructed from measurements of $n$ copies of the state vector  $U(\vt_{k,i})|\vec{0}\rangle$. As a result, by sampling subsets of commuting local terms in each parameter update step, which can be measured simultaneously on the state prepared by a single circuit execution, we are able to decrease the variance of the estimators while conserving the number of circuit executions required per parameter update. This approach allows for previously developed methods for grouping commuting Hamiltonian terms \cite{gokhale2019minimizing}, again reminiscent of methods of grouping of terms in digital quantum simulation \cite{Wiebe}, to be easily applied within this context.

Finally, the function \textproc{PartialDerivativeEstimator\_3} in Algorithm \ref{alg:VQE_estimators} makes clear how in a similar manner one could additionally sample over terms (or subsets of terms) of the parameter shift summation, resulting in an SGD algorithm requiring only $n$ measurements per parameter update. Of course, despite the increase in efficiency per optimization step, using small values of $n$ and sampling over terms of the linear combinations leads to estimators with increased variance, and intuitively one may expect this to lead to slower convergence (in terms of number of update steps required), and to solutions which may be relatively far from global minima. This intuition is indeed valid, and in Sections \ref{s:convergence} and \ref{s:results} we explore these trade-offs between efficiency and accuracy, from both an analytical and numerical perspective.

\begin{algorithm}[H]
  \caption{$n$-shot stochastic partial derivative estimators for VQE
    \label{alg:VQE_estimators}}
  \begin{algorithmic}[1]
  	\Statex Given $U(\vt)$ satisfying a $K$-term parameter shift rule, with $\vt \in \mathbb{R}^d$, and $\mathcal{L}(\vt) = \langle H\rangle_{\vt}$ where $H = \sum_{j = 1}^M H_j$
    \Statex
    \Function{PartialDerivativeEstimator\_1}{$i,\vt^{(t)}$}
\ForAll { $1 \leq j \leq M$} \Comment{Iterate through Hamiltonian terms}
        		\ForAll { $1 \leq k \leq K$} \Comment{Iterate through parameter shift terms}
        			\State Evaluate $\tilde{h}_j^{(n)}(\vt^{(t)}_{k,i})$  \Comment{Construct estimator for $\langle H_j \rangle_{\vt^{(t)}_{k,i}}$ via $n$ measurements}
        		\EndFor
        	\EndFor
        \State $g_i(\vt^{(t)}) \gets\sum_{j = 1}^M\sum_{k = 1}^K\gamma_{k,i}\tilde{h}_j^{(n)}(\vt^{(t)}_{k,i})$  \Comment{Construct estimator for $\partial \langle H \rangle_{\vt^{(t)}}/\partial \theta_i $ via linear combination}
    \State \Return $g_i(\vt^{(t)})$
    \EndFunction
            \Statex 
        \Function{PartialDerivativeEstimator\_2}{$i,\vt^{(t)}$}
    \State Sample $j \in \{1,\ldots,M\}$ with $\mathrm{prob}(j) = 1/M$ \Comment{Sample a Hamiltonian term}
        		\ForAll { $1 \leq k \leq K$} \Comment{Iterate through parameter shift terms}
        			\State Evaluate $\tilde{h}_j^{(n)}(\vt^{(t)}_{k,i})$  \Comment{Construct estimator for $\langle H_j \rangle_{\vt^{(t)}_{k,i}}$ via $n$ measurements}
        		\EndFor
        \State $g_i(\vt^{(t)}) \gets \sum_{k = 1}^K M\gamma_{k,i}\tilde{h}_j^{(n)}(\vt^{(t)}_{k,i})$  \Comment{Construct estimator for $\partial \langle H \rangle_{\vt^{(t)}}/\partial \theta_i $ via linear combination and correction factor}
    \State \Return $g_i(\vt^{(t)})$
    \EndFunction
    \Statex 
        \Function{PartialDerivativeEstimator\_3}{$i,\vt^{(t)}$}
    \State Sample $j \in \{1,\ldots,M\}$ with $\mathrm{prob}(j) = 1/M$ \Comment{Sample a Hamiltonian term}
    \State Sample $k \in \{1,\ldots,K\}$ with $\mathrm{prob}(k) = 1/K$ \Comment{Sample a parameter-shift term}
        			\State Evaluate $\tilde{h}_j^{(n)}(\vt^{(t)}_{k,i})$  \Comment{Construct estimator for $\langle H_j \rangle_{\vt^{(t)}_{k,i}}$ via $n$ measurements}
        \State $g_i(\vt^{(t)}) \gets K M\gamma_{k,i}\tilde{h}_j^{(n)}(\vt^{(t)}_{k,i})$  \Comment{Construct estimator for $\partial \langle H \rangle_{\vt^{(t)}}/\partial \theta_i$ by applying correction factor}
    \State \Return $g_i(\vt^{(t)})$
    \EndFunction
  \end{algorithmic}
\end{algorithm}

\section{Unbiased estimators for QAOA}\label{s:qaoa}

As in the VQE setting, given a problem Hamiltonian 
\begin{equation}
    H^{P} = \sum_{j = 1}^MH_j^{P},
\end{equation}
whose ground state encodes the solution to some combinatorial problem, the QAOA loss function is once again given simply by $\mathcal{L}(\vt) = \langle H^{P}\rangle_{\vt}$. Unlike VQE however, the QAOA algorithm is defined additionally by a specific parameterized circuit architecture, designed such that the variational circuit implements a discretized adiabatic evolution, from some known initial state, into the desired target state. While the loss function is the same, the nature of this specific ansatz has consequences for the design of SGD optimizers, which we focus our attention on here. To be more precise, when using parameterized circuits in the context of hybrid quantum-classical optimization, one often considers parameterized circuits of the form

\begin{equation}
    U(\vt) = \prod_{i = 1}^dU_i(\theta_i),
\end{equation}
where each parameterized gate is parameterized by a distinct parameter \cite{benedetti2019parameterized}, and for models of this type parameter shift rules can be derived under various different assumptions on the form of the constituent gates \cite{schuld2019evaluating} (as per Example~\ref{example:ps}). By contrast, the QAOA parameterized circuit ansatz is defined as
\begin{equation}
    U(\vt) = [e^{-i\theta_dH^B}e^{-i\theta_{d-1}H^P}] \ldots    [e^{-i\theta_2H^B}e^{-i\theta_1H^P}],
\end{equation}
where $H^B = \sum_{j = 1}^{M'}H^B_j$ is a mixing Hamiltonian whose ground state is easy to prepare. Note that in this case the variational parameters define a Trotterized time evolution, and therefore a discretization of a typical adiabatic protocol. Under the assumption that all local terms of both the mixing and problem Hamiltonian commute -- i.e., $[H_j^P,H_{j'}^P] = 0$ and $[H_{j}^B,H_{j'}^B] = 0$ for all $j,j'$ -- one has that
\begin{equation}\label{eq:qaoa_ansatz}
    U(\vt) = \big[\big(\prod_{j' = 1}^{M'}e^{-i\theta_dH^B_{j'}}\big)\big(\prod_{j = 1}^{M}e^{-i\theta_{d-1}H^P_{j}}\big)\big] \ldots    \big[\big(\prod_{j' = 1}^{M'}e^{-i\theta_2H^B_{j'}}\big)\big(\prod_{j = 1}^{M}e^{-i\theta_{1}H^P_{j}}\big)\big],
\end{equation}
and one sees that multiple constituent gates are parameterized by the same variational parameter. In order to understand how a parameter shift rule can still be applied in this case, let us consider first a simple parameterized circuit architecture consisting of only a single variational parameter, which parameterizes a single variational gate, i.e., 
\begin{equation}
    U(\theta) = U_Le^{-i\theta G}U_R.
\end{equation}
In this case we see that for any observable $O$ the partial derivative of the parameterized expectation value is given by 
\begin{equation}\label{e:deriv_simple}
    \frac{\partial}{\partial \theta}\langle O\rangle_{\theta} = \langle \vec{0}| U^\dagger(\theta) O U_L\left(\frac{\partial}{\partial \theta}e^{-i\theta G}\right)U_R|\vec{0}\rangle + \hc,
\end{equation}
and from this expression, if $\mathrm{e}^{-\mathrm{i}\theta G}$ admits a $K(G)$-term parameter shift rule, then by definition
\begin{equation}
    \frac{\partial}{\partial \theta}\langle O\rangle_{\theta} = \sum_{k = 1}^{K(G)}\gamma_k(G)\langle O\rangle_{\theta_{k,G}},
\end{equation}
where we have denoted explicitly that the number of terms, linear coefficients and parameter shifts all depend on $G$. Let us now consider a parameterized circuit with a single variational parameter, parameterizing multiple gates, whose generators $G_j$ all commute, i.e.,
\begin{equation}
     U(\theta) = U_L\big(\prod_j^M e^{-i\theta G_j}\big)U_R.
\end{equation}
In particular, we note by comparison with Eq. \eqref{eq:qaoa_ansatz} that the QAOA ansatz is precisely of this form. In this case, the partial derivative $\langle O\rangle_{\theta}$ is given by 
\begin{align}
    \frac{\partial}{\partial \theta}\langle O\rangle_{\theta} &= \langle \vec{0}| U^\dagger(\theta) O U_L\Big(\frac{\partial}{\partial \theta}\big[\prod_{j'}^Me^{-i\theta G_{j'}}\big]\Big) U_R | \vec{0}\rangle + \hc\\
    &= \sum_{j = 1}^M \langle \vec{0}| U^\dagger(\theta) O U_L \left(\frac{\partial}{\partial \theta}e^{-i\theta G_j}\right)\big[\prod_{j'\neq j}^M e^{-i\theta G_{j'}}\big]U_R|\vec{0}\rangle + \hc\\
    &= \sum_{j = 1}^M \big( \langle \vec{0}| U^\dagger(\theta) O U_L\left(\frac{\partial}{\partial \theta}e^{-i\theta G_j}\right)\tilde{U}_R|\vec{0}\rangle  + \hc\big) .
\end{align}
By comparison with Eq.~\eqref{e:deriv_simple} we therefore see that, under the assumption that all the gates $\mathrm{e}^{-\mathrm{i}\theta G_j}$'s still satisfy a parameter shift rule,
\begin{equation}
    \frac{\partial}{\partial \theta}\langle O\rangle_{\theta} = \sum_{j = 1}^M\sum_{k = 1}^{K(G_j)}\gamma_k(G_j)\langle O\rangle_{\theta_{k,G_j}}.
\end{equation}
At this stage we can finally see that in the case where multiple gates are parameterized by the same variational parameter $\theta$ -- as is the case for the QAOA ansatz --  the parameter shift rule we obtain for the partial derivative $\partial \langle O\rangle_{\theta} /\partial \theta$ can be written as a double sum. The first of these summations runs over all gates parameterized by this parameter, while the second runs over the parameter shift terms for a specific gate. As the QAOA loss function is identical to the VQE loss function, all the unbiased estimators of the previous section can be straightforwardly adapted to the QAOA setting, while the additional structure of the parameter shift rule, allows one to further sample from the linear combination over gates parameterized by the same variational parameter.

\section{Unbiased estimators for MSE Quantum classifiers}\label{s:mse}

Given the tools developed in the previous sections it is now relatively straightforward to construct unbiased estimators, and therefore stochastic gradient descent optimizers, for mean squared error (MSE) quantum classifiers. As we will see, the primary obstacle in this case is that the required partial derivatives are not linear combinations of expectation values, but linear combinations of non-linear functions of expectation values. To begin, it is helpful to define the mean squared error loss, as used for example in univariate regression. Specifically, given a target value $y \in \mathbb{R}$ we define the mean squared error (MSE) loss with respect to the operator $O$ as

\begin{equation}\label{e:mse_loss}
    \ml_{\mathrm{MSE}}(\vt) = (\langle O\rangle_{\vt} - y)^2.
\end{equation}
If one is given a data-set of binary labelled examples, as is the case in classification problems, it is possible to use the average MSE loss over the data-set as the loss function for a classification algorithm. More precisely, given a data-set $\mathcal{D} = \{(x_j,y_j)\, |\, j \in [M]\}$, with $y_j \in \mathbb{R}$, let us define $\langle O\rangle_{\vt,x_j}$ as the expectation value of the operator $O$ with respect to the state vector $U(\vt,x_j)|\vec{0}\rangle$, i.e.,

\begin{equation}\label{e:input_param_circuit}
    \langle O\rangle_{\vt,x_j} = \langle \vec{0}|U^\dagger(\vt,x_j)OU(\vt,x_j)|\vec{0}\rangle,
\end{equation}
where in this case $U(\vt,x_j)$ is the unitary operator corresponding to a quantum circuit parameterized by both $\vt$ and $x_j$. The MSE loss with respect to $O$, over the data-set $\mathcal{D}$, is then defined as 

\begin{align}\label{e:mse_loss_dset}
    \ml(\vt) &= \frac{1}{M}\sum_{j = 1}^M(\langle O\rangle_{\vt,x_j} - y_j)^2 \\
    &:=\frac{1}{M}\sum_{j = 1}^M\ml_{\mathrm{MSE},j}(\vt). \label{e:mse_data-set_loss}
\end{align}
We note that typically one may also add a regularization term to the above loss function, which we have omitted for notational clarity, but can be straightforwardly included in our analysis, as will be later shown. In order to construct stochastic gradient descent optimizers for the loss function of Eq.~\eqref{e:mse_loss_dset}, it is convenient to start with the simplified ``single instance" MSE loss function of Eq.~\eqref{e:mse_loss}. For this loss function, we have that

\begin{equation}
    \frac{\partial \ml_{\mathrm{MSE}}(\vt)}{\partial \theta_i} = 2\langle O\rangle_{\vt}\frac{\partial \langle O\rangle_{\vt}}{\partial \theta_i}  - 2y\frac{\partial \langle O\rangle_{\vt}}{\partial \theta_i}.
\end{equation}
In order to obtain an unbiased estimator for the partial derivative $\partial \ml_{\mathrm{MSE}}(\vt) /\partial \theta_i$  it is sufficient to have \emph{independent} unbiased estimators for $\langle O\rangle_{\vt}$ and $\partial\langle O\rangle_{\vt} /\partial \theta_i$. However, for settings in which a parameter shift rule applies, we have already constructed such estimators. Explicitly, the $n$-sample mean $\tilde{o}^{(n)}(\vt)$ is an unbiased estimator for $\langle O\rangle_{\vt}$ , while under the assumption of a $K$-term parameter shift rule

\begin{equation}
    \tilde{d}_i^{(n)}(\vt) = \sum_{k = 1}^K\gamma_{k,i}\tilde{o}^{(n)}(\vt_{k,i}),
\end{equation}
is an unbiased estimator for $\partial\langle O\rangle_{\vt} /\partial \theta_i$, which is independent from $\tilde{o}^{(n)}(\vt)$ since they are evaluated from measurements of different circuit evaluations. As a result, one can straightforwardly show that the random variable

\begin{equation}
    \tilde{d}^{(n)}_{\mathrm{MSE},i}(\vt) :=  2\tilde{o}^{(n)}(\vt)\tilde{d}_i^{(n)}(\vt)  - 2y\tilde{d}_i^{(n)}(\vt)
\end{equation}
is an unbiased estimator for $\partial \ml_{\mathrm{MSE}}(\vt) /\partial \theta_i$, which requires $(K+1)n$ measurements to construct. Of course, in the same spirit as previous sections, one can also sample single terms from the summation over the parameter shift terms, giving rise to an estimator which requires only $2n$ measurements to evaluate. In light of the above it is now straightforward to design unbiased estimators for MSE quantum classifiers, as the loss function of Eq.~\eqref{e:mse_loss_dset} involves only an additional linear combination over data-set instances. In order to formalize this, let's start with the following definitions:

\begin{definition}
Given a data-set $\mathcal{D} = \{(\vec{x}_j,y_j)\, |\, j \in [M]\}$, let us define
$\tilde{o}^{(n)}(\vt,j)$ as the $n$-sample mean estimator of $\langle O\rangle_{\vt,x_j} = \langle \vec{0}|U^\dagger(\vt,x_j)OU(\vt,x_j)|\vec{0}\rangle$. Furthermore, let $U(\vt, x_j)$ satisfy a $K$-term parameter shift rule (w.r.t. $\vt$), and let us further define
\begin{align}
\tilde{d}_i^{(n)}(\vt,j) &:= \sum_{k = 1}^K\gamma_{k,i}\tilde{o}^{(n)}(\vt_{k,i},j),\\
\tilde{d}^{(n)}_{\mathrm{MSE},i}(\vt,j) &:=  2\tilde{o}^{(n)}(\vt,j)\tilde{d}_i^{(n)}(\vt,j)  - 2y_j\tilde{d}_i^{(n)}(\vt,j).
\end{align}
\end{definition}
\noindent Given these estimators, we have that for the data-set MSE loss function of Eq.~\eqref{e:mse_data-set_loss}
\begin{equation}
\frac{\partial \mathcal{L}(\vt)}{\partial \theta_i}  = \frac{1}{M}\sum_{j = 1}^M \frac{\partial \mathcal{L}_{\mathrm{MSE},j}(\vt)}{\partial \theta_i}, 
\end{equation}
and therefore the random variable 
\begin{equation}
    g^{(n)}_i(\vt) = \sum_{j = 1}^M\frac{1}{M}\tilde{d}^{(n)}_{\mathrm{MSE},i}(\vt,j)
\end{equation}
is an unbiased estimator for $\partial \mathcal{L}(\vt)/ \partial \theta_i$, which requires $M(K+1)n$ measurements to construct. However, as in the case of classical mini-batch SGD, one can sample subsets  $B$ (batches) of data instances in each optimization step, giving rise to a ``doubly stochastic" unbiased estimator which requires only $\lvert B\rvert(K+1)n$ measurements per parameter update. We present this particular algorithm below, for the case $|B| = 1$, however given the explicit VQE estimators of Algortithm \ref{alg:VQE_estimators} it should be clear how this can be extended to include sampling over parameter shift terms.

\begin{algorithm}[H]
  \caption{$n$-shot doubly stochastic partial derivative estimator for MSE type loss functions
    \label{alg:MSE_loss functions}}
  \begin{algorithmic}[1]
  	\Statex Given $U(\vt)$ satisfying a $K$-term parameter shift rule, with $\vt \in \mathbb{R}^d$, a data-set $\mathcal{D} = \{(x_j,y_j)\, |\, j\in[M]\}$, and the loss function of Eq. \eqref{e:mse_data-set_loss}
    \Statex
    \Function{PartialDerivativeEstimator\_4}{$i,\vt^{(t)}$}
\State Sample $j$ uniformly from $\{1,\ldots,M\}$  \Comment{Sample a data-set instance}
        \State Evaluate $\tilde{o}^{(n)}(\vt^{(t)},j)$ \Comment{Construct estimator for $\langle O\rangle_{\vt^{(t)},x_j}$ via $n$ measurements}
        	\ForAll{$1 \leq k \leq K$}    \Comment{Iterate over parameter shift terms}
        		\State Evaluate $\tilde{o}^{(n)}(\vt^{(t)}_{k,i},j)$ \Comment{Construct estimator for $\langle O \rangle_{\vt^{(t)}_{k,i},x_j}$ via $n$ measurements}
        	\EndFor
        	\State $\tilde{d}_i^{(n)}(\vt^{(t)},j) \gets \sum_{k = 1}^K\gamma_{k,i}\tilde{o}^{(n)}(\vt^{(t)}_{k,i},j)$ \Comment{Construct estimator for $\partial \langle O\rangle_{\vt^{(t)},x_j}/\partial \theta_i$}
        \State $\tilde{d}^{(n)}_{\mathrm{MSE},i}(\vt^{(t)},j) \gets 2\tilde{o}^{(n)}(\vt^{(t)},j)\tilde{d}_i^{(n)}(\vt^{(t)},j)  - 2 y_j\tilde{d}_i^{(n)}(\vt,j)$ \Comment{Construct estimator for $\partial \ml_{\mathrm{MSE},j}(\vt^{(t)})/\partial \theta_i$}
        \State $\tilde{g}_i^{(n)}(\vt^{(t)}) \gets \tilde{d}^{(n)}_{\mathrm{MSE},i}(\vt^{(t)},j)$ \Comment{Construct estimator for $\partial \ml(\vt^{(t)})/\partial \theta_i$}
    \State \Return $g_i(\vt^{(t)})$
    \EndFunction
  \end{algorithmic}
\end{algorithm}

\section{Extensions to Generic Loss Functions}\label{s:extensions}

In light of the optimizers presented for the MSE quantum classifier, a natural question is whether or not analogous optimizers can be obtained when regularization terms are included, or when alternative loss functions, such as the cross-entropy, are used. Let us begin with a remark on the inclusion of regularization terms, which can be handled straightforwardly via the following observation:

\begin{observation}
Given some loss function $\ml$, let the random variable $g(\vt)$ be an unbiased estimator for $\nabla \ml(\vt)$ -- i.e., $\mathbb{E}[{g}(\vt)] = \nabla \ml(\vt)$. Then, the random variable 
\begin{equation}
    \tilde{g}(\vt) = g(\vt) + \nabla R(\vt)
\end{equation}
is an unbiased estimator for the regularized loss function $\tilde{\ml}(\vt) = \ml(\vt) + R(\vt)$. This can be straightforwardly verified via
\begin{equation}
\mathbb{E}[\tilde{g}(\vt)] = \mathbb{E}[g(\vt) + \nabla R(\vt)] = \mathbb{E}[g(\vt)] + \nabla R(\vt) =  \nabla \ml(\vt) + \nabla R(\vt) = \nabla \tilde{\ml}(\vt).
\end{equation}
\end{observation}

\noindent Unfortunately however an extension to generic loss functions is not as straightforward. To see this let's consider a loss function $\mathcal{L}$, for which the partial derivative $\partial \mathcal{L}(\vt)/\partial \theta_i$ that we would like to estimate is a non-linear function of an expectation value, i.e., $\partial \mathcal{L}(\vt)/\partial \theta_i= f(\langle O\rangle_{\vt})$ for some non-linear function $f$. Assume now that the random variable $X(\vt)$ is an unbiased estimator for $\langle O \rangle_{\vt}$ (such as the $n$-sample mean estimator). If $f$ had been a linear function, then we would have that
\begin{equation}
    \mathbb{E} [ f(X(\vt)) ] = f(\mathbb{E}[X(\vt)]) = f(\langle O\rangle_{\vt}) = \frac{\partial}{\partial \theta_i}\ml(\vt),
\end{equation}
and as such $f(X)$ would be an unbiased estimator for $\partial \mathcal{L}(\vt)/\partial \theta_i= f(\langle O\rangle_{\vt})$. Unfortunately however for generic non-linear functions $f$ we have that  $f(\mathbb{E}[X]) \neq \mathbb{E} [ f(X) ]$, and as a result unbiased estimators for $f(\langle O\rangle_{\vt})$ cannot be straightforwardly constructed from unbiased estimators for $\langle O\rangle_{\vt}$, and more sophisticated strategies are necessary. While we leave the completely general case open for future work, we illustrate here a method for the construction of unbiased estimators for \emph{polynomial} functions of random variables (of which the previously considered MSE loss function is a special case). In particular, let $f:\mathbb{R}\rightarrow \mathbb{R}$ be a degree $k$ polynomial, i.e. 
\begin{equation}
    f(x) = \sum_{j = 0}^ka_jx^j,
\end{equation}
and let us define the function $h:\mathbb{R}^k \rightarrow \mathbb{R}$ via
\begin{equation}
    h(x_1,\ldots,x_k) = a_0 + \sum_{j = 1}^ka_j\left(\prod_{i = 1}^jx_i\right).
\end{equation}
In addition, for all $m \geq k$, let us denote by $\textbf{P}_{k,m}$ the set of all $m!/(m-k)!$ ordered arrangements $(i_1,\ldots,i_k)$ of size $k$ chosen from $(1,\ldots,m)$. Now, given a set $\{X_j\, |\, j \in [m]\}$ of $m \geq k$ independent copies of the random variable $X$, one can easily verify that the random variable 
\begin{equation}
    g^{(m)}(X_1,\ldots,X_m) = \frac{(m-k)!}{m!}\sum_{\textbf{P}_{k,m}}h(X_{i_1},\ldots,X_{i_k})
\end{equation}
is an unbiased estimator for $f(\mathbb{E}(X))$, which requires $m$ samples from $X$ to construct. In fact, $g^{(m)}$ is the U-statistic for the kernel $h$, and as such is the minimum variance unbiased estimator for $f(\mathbb{E}(X))$ \cite{lee2019u}. As many relevant loss functions, such as the cross entropy or the log-likelihood, are certainly not polynomials, it is worth noting before continuing that despite the absence of rigorous constructions it is still possible to apply the spirit of the constructions from the previous sections -- i.e., utilizing single shot estimations of observables and sampling over linear combinations -- to obtain \emph{biased estimators} for arbitrary loss functions, whose performance on practical tasks can be easily evaluated, and for which one may still be able to prove convergence bounds \cite{biased}. In fact, while unbiased estimators facilitate convergence proofs in simplified settings, it is not a-priori clear that naive (potentially biased) estimators would perform badly as heuristics. Such heuristic analysis, for loss functions such as the cross entropy, should be considered in parallel with a development of the theory. Our motivation in this work is to provide a general framework, and to lay the foundations for more sophisticated analysis.

\section{Convergence Guarantees}\label{s:convergence}

The construction of convergence guarantees for stochastic gradient descent optimizers in fully realistic settings is currently an active and open area of research, of critical importance for building a theoretically solid foundation for state of the art machine learning algorithms \cite{nguyen2018tight}. In particular, as discussed in the previous sections, many prior approaches to gradient descent optimization in the context of hybrid quantum-classical optimization have been large $n$ instances of the algorithms from the previous sections (without sampling over linear combinations), and one of the primary motivations of this work is to place these previously heuristic approaches on a solid formal footing. While convergence guarantees are not yet available for state of the art models, highly non-convex landscapes and optimization algorithms with heuristically adapted learning rates, in certain simplified settings it is indeed possible to obtain convergence theorems, which allow one to build a measure of confidence in these techniques. In order to gain a more formal perspective on the algorithms presented in the previous sections, as well as insight into the requirements for the development of further well motivated stochastic gradient descent algorithms, we examine in this section convergence guarantees for loss functions  $\mathcal{L}:\mathbb{R}^d \rightarrow \mathbb{R}$ which satisfy the \emph{Polyak-\! Lojasiewicz} (PL) inequality, a slightly generalized notion of strong convexity \cite{karimi2016linear}. These bounds are complimentary to the bounds previously discussed by Harrow and Napp for strongly-convex loss functions \cite{harrow_napp}, and are presented explicitly in order to facilitate discussion around both the technical requirements for such guarantees in the hybrid quantum-classical setting, and the trade-offs inherent in SGD optimization schemes, which are explored numerically in the following section. In particular, a loss function  $\mathcal{L}:\mathbb{R}^d \rightarrow \mathbb{R}$ satisfies the PL inequality for some $\mu > 0$ if for all $\vt \in \mathbb{R}^d$ it is true that

\begin{equation}\label{e:PL_inequality}
\frac{1}{2} \nl\nabla \ml(\vt)\nr^2 \geq \mu \big(\ml(\vt) - \ml(\vt^*)\big),
\end{equation}
where $\vt^* \in \mathbb{R}^d$ is the value at which $\ml(\vt)$ attains a global minimum. Under the assumption of the PL inequality, one can state the following convergence theorem, whose proof can be found in Refs.~\cite{wolf2018mathematical, karimi2016linear}:

\begin{theorem}[Stochastic gradient descent convergence \cite{wolf2018mathematical}]\label{t:SGD} Let $\ml \in C^1(\mathbb{R}^d)$ satisfy the Polyak-Lojasiewicz inequality in Eq.~\eqref{e:PL_inequality} for some $\mu  > 0$, have $L$-Lipschitz gradient and a global minimum attained at $\vt^*$. For any $T\in \mathbb{N}$ and any $\vt \in \mathbb{R}^d$ let $g^{(1)}(\vt), \ldots, g^{(T)}(\vt)$ be i.i.d random variables with values in $\mathbb{R}^d$ such that $\mathbb{E}[g^{(t)}(\vt)] = \nabla \ml(\vt)$. With $\vt^{(0)} \in \mathbb{R}^d$ consider the sequence $\vt^{(t+1)} := \vt^{(t)} - \alpha g^{(t)}(\vt^{(t)})$.
\begin{enumerate}
    \item If $\forall \vt, t:\mathbb{E}[\nl g^{(t)} (\vt)\nr^2] \leq \gamma^2$ and $\alpha \in [0, 1/(2\mu)]$, then
    \begin{equation}\label{e:conv_distance_1}
        \mathbb{E}[\ml(\vt^{(T)})] - \ml(\vt^*) \leq (1-2\mu\alpha)^T\big(\ml(\vt^{(0)}) - \ml(\vt^*)\big) + \frac{L\alpha\gamma^2}{4\mu}.
    \end{equation}
    \item If $\forall \vt, t:\mathbb{E}[\nl g^{(t)} (\vt)\nr^2] \leq \beta^2 \nl\nabla\ml(\vt)\nr^2$ and $\alpha = 1/(L\beta^2)$, then
    \begin{equation}\label{e:conv_distance_2}
        \mathbb{E}[\ml(\vt^{(T)})] - \ml(\vt^*) \leq \Big(1 - \frac{\mu}{L\beta^2} \Big)^T\big(\ml(\vt^{(0)}) - \ml(\vt^*)\big).
    \end{equation}
\end{enumerate}
\end{theorem}
\noindent We note that as all previously discussed algorithms have been constructed via unbiased estimators, Theorem \ref{t:SGD} applies directly to these algorithms provided one can prove Lipschitz continuity of the gradient $\nabla\mathcal{L}$. 
Additionally, in order to apply the convergence bounds of Theorem \ref{t:SGD} in a quantitative manner, it is necessary to obtain an upper bound on the quantity $\mathbb{E}[\nl g^{(t)} (\vt)\nr^2]$, which in principle amounts to calculating the variance of the estimator $g^{(t)}(\vt)$, as can be seen from the following straightforward calculation:
\begin{align}
    \mathbb{E}[\nl g^{(t)} (\vt)\nr^2] &= 
    \mathbb{E}[\nl g^{(t)} (\vt)\nr^2] - \mathbb{E}[\nl g^{(t)} (\vt)\nr]^2 + \mathbb{E}[\nl g^{(t)} (\vt)\nr]^2 \\
    &= \operatorname{Var}[ g^{(t)} (\vt)] + \nl \nabla \mathcal{L}(\vt)\nr^2,
\end{align}
where we defined the variance of a random vector as $\operatorname{Var}[X] := \mathbb{E}[\nl X \nr^2] - \mathbb{E}[\nl X\nr]^2$. As Theorem \ref{t:SGD} only applies in the simplified setting of loss functions which satisfy the PL inequality we will not address these upper bounds explicitly here, as in practice more sophisticated convergence theorems are necessary to make applicable quantitative statements for realistic settings. However, we note the critical role this variance plays in determining the distance between the global optimum and the solution one can be guaranteed to converge to, and in Section \ref{s:results} we explore this interplay, in realistic settings, through extensive numerical experiments. As such, we focus on the issue of Lipschitz continuity, and provide the following results proving that for a large class of realistic parameterized quantum circuits, under the assumption of a parameter shift rule, the gradient of the loss functions considered in this work are indeed Lipschitz continuous. To be specific, we start with the following theorem (whose proof can be found in Appendix \ref{app:lipschitz_proof}), which guarantees Lipschitz continuity of expectation values for a specific class of parameterized quantum circuits:

\begin{theorem}[Lipschitz continuity of expectation values]\label{thm:lips_cont1}
Consider a function $f:[0,2\pi]^M\mapsto\mathbb{R}$ defined by $f(\vt)=\langle\vec{0}|U^\dagger(\vt)OU(\vt)|\vec{0}\rangle$, where $|\vec{0}\rangle$ is an arbitrary state in $\mathbb{C}^D$ for some finite $D$, $U(\vt)$ is a quantum circuit consisting of an arbitrary finite number of fixed gates, and $M$ parameterized gates $U_j(\theta_j) = \mathrm{exp}(-i(\theta_j + c_j)H_j)$, with $H_j$ Hermitian. For any observable $O$ and any set of Hermitian operators $\{H_j\, |\, j\in[M]\}$, the function $f(\vt)$ is $L$-Lipschitz with $L = \sqrt{M}[\max_j(\sup_{\vt}( \vert\partial f(\vt)/\partial \theta_j\vert))]$.
\end{theorem}

\noindent Given this result, an upper bound to $L$ can be obtained as follows: One finds
\begin{eqnarray}
    L &= &\sqrt{M}\max_j \sup_{\vt} \left\vert 2 \mathrm{Re}\left[ \bra{\vec{0}}
        U^\dagger(\vt) O \frac{\partial U(\vt)}{\partial \theta_j}
    \ket{\vec{0}} \right] \right\vert \\
    &\leq  & 2\sqrt{M}\,\max_j \sup_{\vt} \left\vert \bra{\vec{0}}
        U^\dagger(\vt) O U_L(\vt)(-\mathrm{i}H_j)U_R(\vt)
    \ket{\vec{0}} \right\vert \label{proof2} \\
    &\leq
        &2\sqrt{M}\,\max_j \sup_{\vt} \left\Vert O U(\vt) \ket{\vec{0}} \right\Vert_2~
        \left\Vert U_L(\vt)H_jU_R(\vt)\ket{\vec{0}} \right\Vert_2 \label{proof3}  \\
    &\leq &2\sqrt{M}\,\max_j \sup_{\vt} 
        \Vert O \Vert_2~ \Vert U(\vt) \ket{\vec{0}} \Vert_2~
        \Vert U_L(\vt)\Vert_2~ \Vert H_j\Vert_2~
        \Vert U_R(\vt)\ket{\vec{0}}\Vert_2 \\
    &=  & 2\sqrt{M}\,\max_j \Vert O \Vert_2~ \Vert H_j \Vert_2, \label{e:bound_ref}
\end{eqnarray}
where $U(\vt) \equiv U_L(\vt)e^{-i\theta_j H_j}U_R(\vt)$, and Eq. \eqref{proof3} was obtained from Eq. \eqref{proof2} via the Cauchy-Schwarz inequality. In order to cover more realistic loss functions, we need to extend Theorem~\ref{thm:lips_cont1} to sums and products of parameterized expectation values. For example, we have already seen that for the MSE loss function, under the assumption of a parameter shift rule, we have that
\begin{equation}
\frac{\partial}{\partial \theta_j}\ml_{\mathrm{MSE}}(\vt) = 2\langle O\rangle_{\vt}\frac{\partial}{\partial \theta_j}\langle O\rangle_{\vt}  - 2y\frac{\partial}{\partial \theta_j}\langle
O\rangle_{\vt} 
\end{equation}
with $\partial\langle O\rangle_{\vt}/ \partial \theta_j$ given by
\begin{equation}\label{e:parameter_shift_1}
    \frac{\partial }{\partial \theta_j}\langle O\rangle_{\vt} =\sum_{k = 1}^K \gamma_{k,j}\langle O\rangle_{\vt_{k,j}},
\end{equation}
where $\vt_{k,j} =  \vt + c_{k,j}\vec{e}_j$. The first thing to notice is that if $\langle O\rangle_{\vt}$ is $L$-Lipschitz continuous via Theorem \ref{thm:lips_cont1}, then so is $\langle O\rangle_{\vt_{k,j}}$. We see this by noting that 
\begin{equation}
    U_j(\theta_j + c_{k,j})=\mathrm{exp}(-i(\theta_j + c_{k,j})H_j) =\mathrm{exp}(-i\theta_j H_j)\mathrm{exp}(-i c_{k,j}H_j)
\end{equation}
for any constant $c_{k,j}$, and therefore $\langle O\rangle_{\vt_{k,j}}$ satisfies the assumptions of the Lemma after absorbing $\mathrm{exp}(-ic_{k,j})H_j$ into the set of fixed (non-parameterized) gates. This observation, in conjunction with standard result that if $f(\vt)$ and $g(\vt)$ are Lipschitz continuous functions bounded on their domains, then $f + g$ and $fg$ are Lipschitz continuous \cite{Lipshitz}, 
guarantees Lipschitz continuity for the class of 
functions we are interested in.

\section{Numerical Experiments and Benchmarks}\label{s:results}

Given the unbiased estimators of Sections \ref{s:vqe}, \ref{s:qaoa} and \ref{s:mse}, we explore here their performance on benchmark tasks. In particular, while it is clear from the above analysis that sampling over linear combination terms and utilizing small values of $n$ leads to algorithms which require significantly fewer measurements per optimization step than previous approaches, while possessing convergence guarantees in simplified settings, it is not a-priori clear how these algorithms perform on realistic tasks. Specifically, as the convergence properties of the algorithms depend both on the Lipschitz constant of the gradient of the loss, and upper bounds on the variance of the estimator $g^{(t)}(\vt)$, it may be that the advantages of both linear combination sampling and small $n$ implementations are outweighed by slow or unstable convergence behavior in practice. Fortunately, as shown in the following sections, at least in the simplified setting of noiseless devices this is not the case, and small $n$ implementations of the previously introduced algorithms indeed offer multiple advantages over large $n$ approaches, which can be further enhanced through the use of heuristics such as decaying learning rate, and momentum based optimizers \cite{Adam_opt}. All numerical results presented in the following sections were obtained using \texttt{qradient} \cite{qradient}, a custom open-source package for the simulation of hybrid quantum-classical optimization, and all scripts, data and random number seeds are available on request.

\subsection{Stochastic gradient descent for VQE}\label{ss:VQE}

\begin{figure}
\begin{center}
      \includegraphics[width=\columnwidth]{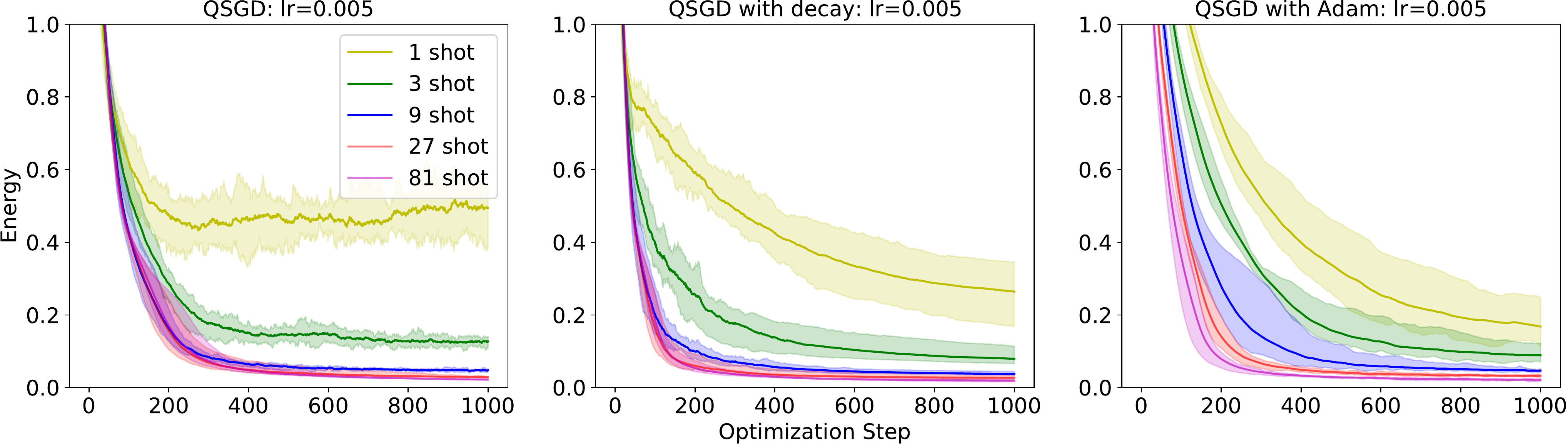}
      \includegraphics[width=\columnwidth]{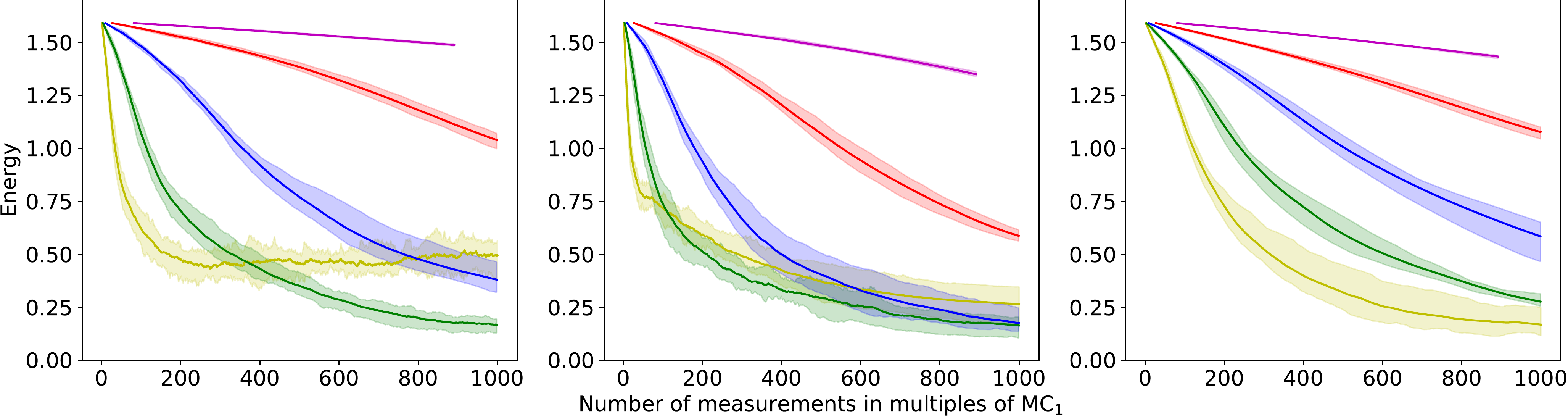}
\end{center}
  \caption{Results of VQE experiments with the estimator \textproc{PartialDerivativeEstimator\_1} from Algorithm \ref{alg:VQE_estimators} (referred to as $n$-shot SGD). The top row of results is plotted with optimization step on the $x$-axis. The lower row of results is plotted with number of measurements on the $x$-axis, in multiples of $\mathrm{MC}_1$, which is the number of measurements required per optimization step by $1$-shot SGD. The left column of figures indicates the results for $n$-shot SGD with fixed learning rate $\alpha = 0.005$. The central column of figures shows the results for $n$-shot SGD with learning rate decay, from an initial learning rate of $\alpha = 0.005$. The right hand column shows the results of  $n$-shot SGD with learning rate decided adaptively by the Adam optimizer \cite{Adam_opt}, starting from an initial learning rate of $\alpha = 0.005$. For each algorithm, and each value of $n$, the experiment has been repeated $8$ times, and the minimum, mean and maximum at each step is displayed. }\label{fig:vqe_results}
\end{figure}

We start by exploring, for different values of $n$, the estimator formalized by the function \textproc{PartialDerivativeEstimator\_1} in Algorithm \ref{alg:VQE_estimators} (which we refer to as $n$-shot SGD), allowing us to gain insight into the effect of efficient expectation value estimation, before introducing the additional stochasticity of Hamiltonian term sampling. In particular, we consider the \emph{critical transverse field Ising model} with open boundary conditions, given by the Hamiltonian
\begin{equation}
H = \sum_{j = 1}^{N-1} Z_jZ_{j+1} + \sum_{j = 1}^{N} X_j,\label{e:vqe_ham}
\end{equation}
with $N=8$, where $X_j$ and $Z_j$ are the respective Pauli operators on site $j$. Critical Ising models are good candidates for such an endeavor, as the entanglement structure of such critical systems prohibits the use of classical tensor network methods for large system sizes \cite{Schuch_MPS}. The parameterized circuit that we consider consists of a sequence of $\sigma$-blocks (with $\sigma \in [X,Y,Z]$), where a single $\sigma$-block has the following structure:

\begin{enumerate}[itemsep=0em]
    \item A layer of parameterized rotations around the $\sigma$ axis -- i.e., an $R_{\sigma}(\theta) = e^{-i\theta/2 \sigma}$ gate acting on all qubits.
    \item A layer of nearest neighbour $\mathrm{CNOT}$ gates with control on qubit $j$ and target on qubit $j+1$, for all even $j$.
    \item A layer of nearest neighbour $\mathrm{CNOT}$ gates with control on qubit $j$ and target on qubit $j+1$, for all odd $j$.
\end{enumerate} 
Note that for one-dimensional circuits consisting of $N$ qubits each block has $N$ free parameters -- i.e., the rotation angles on each qubit. The particular architecture that we consider consists of an initial $Y$-block, with all free parameters set to $\pi/4$, followed by $k$ parameterized $\sigma$-blocks, alternating between $X,Y$ and $Z$ blocks, in that order. A $k$ block circuit of this type therefore has $kN$ free parameters, and we consider a circuit with $50$ blocks, and therefore $400$ free parameters, as $N=8$. An explicit circuit diagram can be found in Appendix \ref{app:parameterized_circuits}. We implemented Algorithm~\ref{alg:stochastic_gradient_base}, calling the estimator \textproc{PartialDerivativeEstimator\_1}, with a fixed learning rate $\alpha = 0.005$, for multiple values of $n$, as well as two heuristic variations:

\begin{enumerate}[itemsep=0em]
    \item \emph{Learning rate decay:} If the energy has not decreased in the last 20 optimization steps, then the learning rate is decreased by a factor of 2.
    \item \emph{Adam optimizer:} A variation of Algorithm \ref{alg:stochastic_gradient_base} which computes adaptive learning rates for each parameter, using the history of prior gradient estimates \cite{Adam_opt}. Implemented with an initial learning rate of $\alpha=0.005$ and hyper-parameters $\beta_1 = 0.9$ and $\beta_2 = 0.999$.
\end{enumerate}
The results of these experiments can be seen in Fig.~\ref{fig:vqe_results}, and there are a variety of interesting aspects to note. Firstly, from the top left figure one can see that, as one might expect, with a \emph{fixed} learning rate the quality of the final solution is strongly effected by the value of $n$. In particular, for very small values of $n$ we expect a large variance in the gradient, even close to a local minima, and the algorithm can therefore only converge to relatively inaccurate solutions. Again as we might expect, the accuracy of the obtained solutions can be successively tightened by increasing the value of $n$, and in the process decreasing the variance of the gradient between optimization steps. By comparing the top left figure with the central and right figure of the top row we can however see that the quality of the final solution can also be significantly improved, even for very small values of $n$, by utilizing adaptive or decaying learning rates, and once again this makes intuitive sense. Once one has converged to a fixed solution, decreasing the learning rate scales the gradient estimator, and therefore effectively decreases the variance in the parameter updates, allowing one to move closer to a local minimum, even with a large variance in the gradient \cite{pmlr-v89-li19c}. As such, we see that both decaying learning rate and increasing shot number allow one to improve the quality of the final solution that is obtained. We note however, that analogously to SGD in purely classical deep learning where variance can be decreased either through scaling via learning rate decay or directly via the batch size, it is not a-priori clear which strategy is optimal \cite{smith2017don}. Additionally, from the top row of results in Fig. \ref{fig:vqe_results} it would appear that even with adaptive learning rate decay, one can achieve better quality solutions by utilizing larger shot numbers, and as such it is not clear that there is something to be gained by using small $n$ estimators.

In order to explore this more clearly, let us define the number of measurements required per optimization step of an $n$-shot algorithm as $\mathrm{MC}_n$ -- where MC stands for \emph{measurement cost}. The bottom row of results in Fig.~\ref{fig:vqe_results} shows the performance of the algorithms not with respect to optimization steps, but with respect to number of measurements performed, as a multiple of $\mathrm{MC}_1$ (the smallest possible number of measurements per optimization step). With respect to this metric, it is clear from the bottom row of figures that in terms of numbers of measurements, single shot ($n=1$) stochastic gradient descent algorithms converge much faster than large $n$ algorithms. 
However, by comparing the upper and lower rows of results one can see that in this setting, for both standard SGD and SGD with decay (left and central panel), in order to improve the quality of the final solution, eventually it is necessary to increase the value of $n$. While the situation is less clear for the Adam optimizer, it also seems likely that increasing $n$ at some point is necessary to achieve higher quality solutions in this case.
As such, a natural strategy, which will be discussed further in Section \ref{s:conclusion}, would be to combine learning rate decay with increasing shot numbers. This should allow one to exploit the rapid convergence, and therefore enhanced measurement efficiency, of small $n$ estimators, while still allowing one to obtain the solutions which are accessible via large $n$ estimators. In fact, recent work has explored precisely such heuristic optimization strategies, with promising results \cite{kubler2019adaptive}.

\begin{figure}
\begin{center}
      \includegraphics[width=.8\columnwidth]{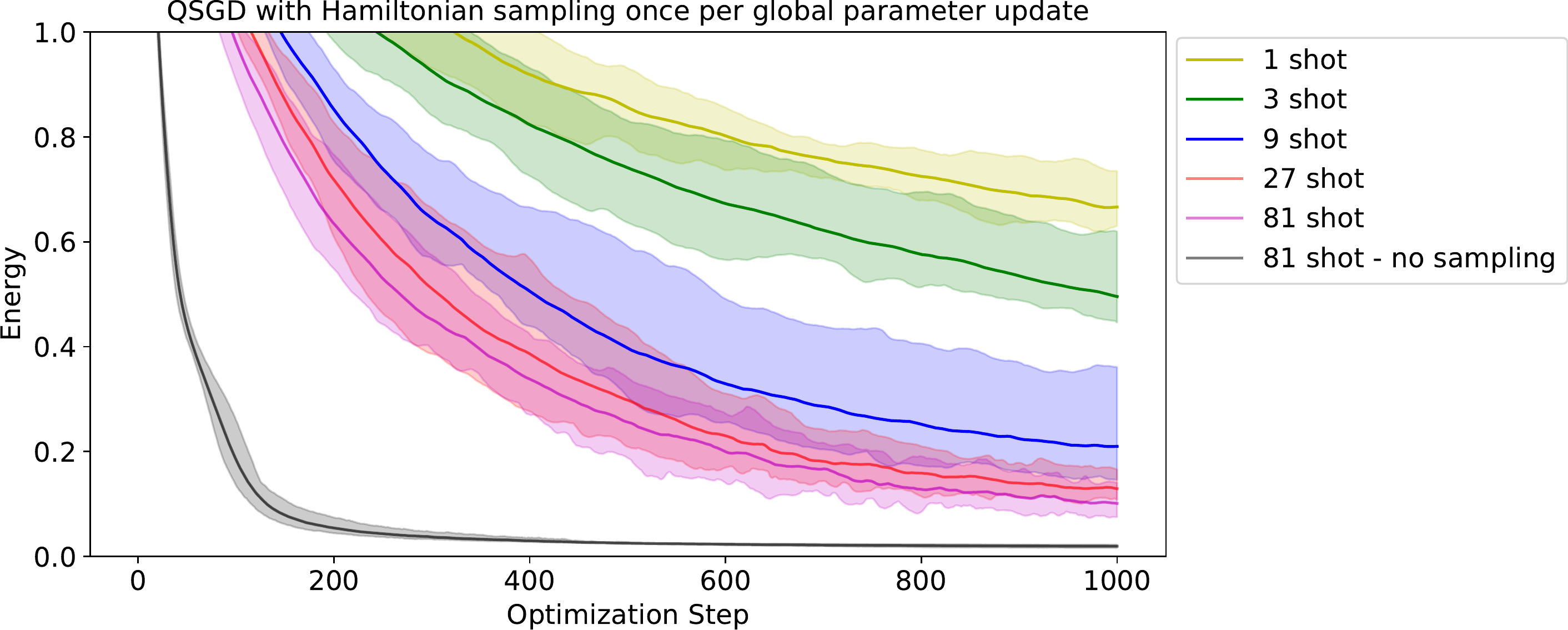}
\end{center}
  \caption{Results of VQE experiments with the estimator \textproc{PartialDerivativeEstimator\_2} from Algorithm \ref{alg:VQE_estimators} -- i.e., $n$-shot stochastic gradient descent incorporating Hamiltonian sampling. The results were obtained using the Adam optimizer, with an initial learning rate of $0.005$. Each experiment has been repeated $8$ times, and the minimum, mean and maximum at each step is displayed. Additionally, the results of 81 shot SGD without Hamiltonian sampling (i.e. using \textproc{PartialDerivativeEstimator\_1}) are shown for comparison.}\label{fig:vqe_results_algs_23}
\end{figure}

Given these insights, we explore in Fig.~\ref{fig:vqe_results_algs_23} the performance of the estimator \textproc{PartialDerivativeEstimator\_2} -- i.e., $n$-shot stochastic gradient descent with Hamiltonian sampling -- in order to explore the additional effect of sampling over linear combinations. As discussed in Section \ref{s:vqe}, in practical settings one would not sample single terms of the Hamiltonian per parameter update, but rather commuting sets of Hamiltonian terms which can be easily measured simultaneously, so that the maximum amount of information can be extracted from a single circuit forward pass. However, for illustration purposes, we explore here the performance of this algorithms when a single local term of the Hamiltonian is sampled, as this is clearly an extreme case, whose performance can only be improved by sampling more terms simultaneously. Once again we implemented Algorithm \ref{alg:stochastic_gradient_base} with learning rate decided by the Adam optimizer, with an initial learning rate of $\alpha = 0.005$. As a result of Theorem \ref{t:SGD}, and from our numerical experiments with \textproc{PartialDerivativeEstimator\_1}, due to the high variance of the estimator in this case, even with large values of $n$, we would expect this algorithm to converge slowly as a function of the number of optimization steps required, and indeed this is what is observed in Fig.~\ref{fig:vqe_results_algs_23}. However, the crucial insight is that although more optimization steps are required, each optimization step of Algorithm \ref{alg:stochastic_gradient_base} requires up to a factor of $M$ less circuit executions (where $M$ is the number of local terms of the Hamiltonian), and using Hamiltonian sampling may well facilitate more rapid convergence or initialization, with respect to the number of circuit executions required. As the precise efficiency gains compared to the simple $n$-shot estimator \textproc{PartialDerivativeEstimator\_1} depend on the strategy via which local Hamiltonian terms are grouped and measured (i.e., $\mathrm{MC}_1$ varies depending on the strategy for obtaining measurements of all local Hamiltonian terms), we have not provided a direct comparison between \textproc{PartialDerivativeEstimator\_1} and \textproc{PartialDerivativeEstimator\_2} in terms of circuit executions and number of measurements, and we leave such comparisons, with state-of-the art strategies for minimizing all involved costs, to future work, emphasizing that this depends on the particular Hamiltonian at hand. Despite this, it is clear from Fig.~\ref{fig:vqe_results_algs_23} that the final solutions obtained when using \textproc{PartialDerivativeEstimator\_2} are relatively far from the optimal solutions obtained by large $n$ versions of  \textproc{PartialDerivativeEstimator\_2}, and as such once again a natural strategy would be to use Hamiltonian sampling as an initialization strategy, with the number of terms (or commuting sets of terms) of the Hamiltonian treated as a hyper-parameter.

\subsection{Stochastic gradient descent for QAOA}\label{ss:QAOA}

\begin{figure}
\begin{center}
      \includegraphics[width=\columnwidth]{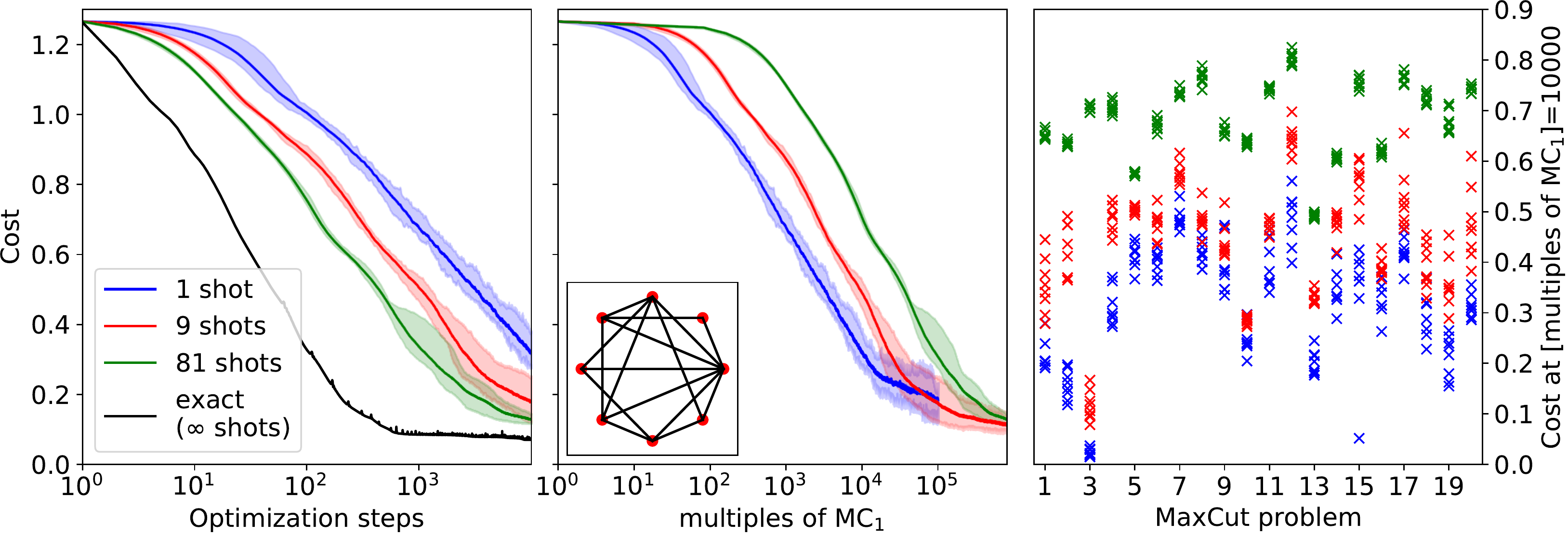}
\end{center}
  \caption{Results of QAOA experiments with $n$-shot SGD (i.e. \textproc{PartialDerivativeEstimator\_1} from Algorithm \ref{alg:VQE_estimators}). The left and center figures show the results obtained for the specific \textsf{MaxCut} instance displayed in the inset. While in the left figure the cost is plotted against the number of optimization steps, in the center figure the cost is plotted against the number of circuit executions, up to a constant factor $MC_1$. The right most figure shows the results obtained by $n$-shot SGD after 10000MC1 measurements, on 20 randomly drawn \textsf{MaxCut} instances, all with $\vert V \vert = 8$ and $\vert E \vert = 16$. As the ground state energy of each randomly drawn problem instance $H^P$ may be different (dependent on the number of maximum possible cuts), the cost refers to the energy $\braket{0 \vert U^\dagger(\theta) H^P U(\theta) \vert 0}$, normalized by the absolute value of the ground state energy for the particular problem instance, and shifted by +1. This rescaling ensures the minimum achievable cost for all problem instances is zero, and allows for meaningful comparison between problem instances. }\label{fig:qaoa_results}
\end{figure}

As discussed in Section \ref{s:qaoa}, if one does not sample over linear combinations of local Hamiltonian terms, or over parameter shift terms, then the unbiased estimator formalized via \textproc{PartialDerivativeEstimator\_1} from Algorithm \ref{alg:VQE_estimators} can be directly applied in the context of the QAOA circuit ansatz, by simply re-interpreting the double sum from the parameter shift rule as a single sum over a multi-index. In this section we explore the performance of this estimator, again for different values of $n$, on various instances of the \textsf{MaxCut} problem \cite{MaxCut}. Specifically, given a graph $G = (V,E)$, the \textsf{MaxCut} problem is equivalent to finding the ground state of the Hamiltonian
\begin{equation}
    \label{eq:qaoa}
    H^P = \sum_{(i,j) \in E} Z_iZ_j.
\end{equation}
As mentioned in Section \ref{s:qaoa}, the QAOA ansatz alternates application of the problem Hamiltonian with a mixing Hamiltonian $H^B$ via 
\begin{equation}
    U(\vt) = [e^{-i\theta_dH^B}e^{-i\theta_{d-1}H^P}] \ldots    [e^{-i\theta_2H^B}e^{-i\theta_1H^P}].
\end{equation}
For all experiments performed here, a value of $d=100$ has been used, with the mixing Hamiltonian $H^B = \sum_{j = 1}^{|V|}X_j.$ Fig.~\ref{fig:qaoa_results} shows the results obtained by using the above QAOA ansatz, in conjunction with the estimator \textproc{PartialDerivativeEstimator\_1}, on 20 randomly drawn instances of the \textsf{MaxCut} problem, all with $|V| =8$ vertices, and $|E| = 16$ edges. For all instances, the parameters were initialized such that they realize a simple linear interpolation between $H^P$ and $H^B$ -- i.e., for odd $j$ (parameters of $H^P$ terms) we have that $\theta_j = j/d$ while for even $j$ (parameters of $H^B$ terms) we have that $\theta_j = 1 - j/d$. This initialization is informed by recent studies which identified global optima for \textsf{MaxCut} problems, which were typically close to linear interpolation solutions \cite{zhou2018quantum}.
The optimization was carried out using the \textit{Adam} optimizer, with initial learning rate $\alpha = 0.001$ and hyper-parameters $\beta_1 = 0.8 $ and $\beta_2 = 0.999$. Once again we are able to draw various insights from the results of Fig.~\ref{fig:qaoa_results}, which are similar in nature to the conclusions from our VQE simulations. In particular, by comparing the left and center panels of Fig.~\ref{fig:qaoa_results}, which show the convergence behavior with respect to optimization steps and measurement cost respectively for one specific \textsf{MaxCut} instance, we once again see that while 1-shot SGD converges slower in terms of the number of optimization steps required, it converges faster in terms of the number of measurements required (where once again we have plotted the number of measurements in multiples of $MC_1$, the measurement cost per optimization step of single shot SGD). The rightmost panel of Fig.~\ref{fig:qaoa_results} confirms that this behavior is indeed generic. Specifically, the right most panel shows the cost obtained by $n$-shot SGD after $10000MC_1$ measurements, on 20 randomly drawn \textsf{MaxCut} instances, and it is clear that for all instances, single shot SGD obtains (often significantly) lower costs with this number of measurements, and hence should clearly be used as an initialization strategy. Given this initialization strategy the second natural question is then whether learning rate decay is sufficient to decrease the variance of the parameter updates such that highly accurate solutions can be obtained, or whether in addition it is necessary to increase the shot number. From the central panel (in which the number of steps is limited by computational resources) the fact that the 9-shot curve intersects the single shot curve suggests that, for this particular example, the 9-shot algorithm is likely to converge to a more accurate solution than the single shot algorithm. As such, if one begins with the single shot algorithm, then it is possible that increasing the number of measurements at some point would allow one to obtain a more accurate solution than one that could be obtained by keeping the number of measurements constant and simply decreasing the learning rate further. Given this, as per the VQE setting, in order to achieve the most accurate solution with the minimum number of measurements, a promising strategy is to combine learning rate decay with an increasing number of measurements.

\subsection{Stochastic gradient descent for MSE classification}

In order to investigate the performance of $n$-shot doubly stochastic gradient descent for MSE loss functions on a realistic task, we consider the problem of image classification. In particular, we consider a binary classification task defined by a data-set consisting of down-sampled grayscale $3$'s and $6$'s from the MNIST data-set. To be specific, each selected MNIST image has been down-sampled from $28\times 28$ pixels to $8\times 8$ pixels by first removing 6 pixels from all boundaries, and then removing every second row and every second column. Each image was then flattened into a 64 dimensional vector, and the pixel values were normalized so that the 2-norm of each image vector was 1. Each image was then encoded into an $8$ qubit state via amplitude encoding \cite{SchuldClassification, schuld2018supervised}. The training data-set consisted of 2000 $3$'s (labelled with $y=1$) and 2000 $6$'s (labeled with $y = -1$), while the validation data-set consisted of 200 $3$'s and $200$ $6$'s. We have utilized the same parameterized circuit architecture as in the previous section (with $6$ qubits and 18 $\sigma$-blocks) followed by a measurement of $Z_1$, where given an image instance $x$ the model has been defined via
\begin{equation}\label{e:model}
f(x) =   \begin{cases} 
   1 & \text{if } \langle Z_1 \rangle \geq 0 \\
   -1       & \text{else }.
  \end{cases}
\end{equation}

\noindent The loss function has been as per Eq.~\eqref{e:mse_data-set_loss}, with $O = Z_1$, and we investigated the performance of the estimator defined in Algorithm \ref{alg:MSE_loss functions} -- resulting in what we refer to as $n$-shot doubly stochastic gradient descent (DSGD) -- for different \emph{fixed} learning rates and values of $n$, as can be seen in Fig.~\ref{fig:classifier_results}. In particular, in order to provide a benchmark the top row of Fig.~\ref{fig:classifier_results} shows the results for ``singly''-stochastic gradient descent -- i.e., mini-batch gradient descent with batch size 1 and \emph{exact} expectation values (which, as discussed, is not practically feasible), while the bottom row of Fig.~\ref{fig:classifier_results} shows the results for $n$-shot doubly stochastic gradient descent, also with batch size 1.  An important thing to note is that because of the way the model is defined in Eq.~(\ref{e:model}), a high confidence (large shot number $n$) estimate of $\langle Z_1\rangle$ is required in order to make a prediction with the model, even though, as we have discussed at length, such an estimate is not required to implement the doubly stochastic optimization algorithm. 

Additionally, we note that in order to avoid overfitting, a typical supervised learning procedure involves dividing the optimization into ``epochs"  -- a single pass of the optimization algorithm through the training data set -- each of which is followed by an evaluation of the optimization procedure, and a subsequent determination of convergence, by calculating the model accuracy on the validation data-set \cite{Goodfellow-et-al-2016}. In light of the prior observations concerning the shot numbers necessary for high-confidence model predictions, a natural training strategy in this setting is therefore to train the algorithm for a single epoch, using low shot number estimates and without making simultaneous predictions on the training data, before then evaluating the accuracy of the model on the (typically smaller) validation set, using large shot number estimates of the expectation value to make predictions.
In order to reflect this strategy, and the information one would then have available in a practical implementation of this algorithm, for the ``singly"-stochastic gradient descent, in which exact expectation values are used in training, we have plotted both the training and validation set accuracy after each epoch, while for the $n$-shot doubly stochastic gradient descent we have plotted only the validation set accuracy, which was evaluated after each training epoch by using exact expectation values (as a proxy for the much higher shot numbers one would use in practice).

Once again, there are a variety of lessons to be taken from the results of Fig.~\ref{fig:classifier_results}. Firstly, as can be seen by comparing the columns, it is clear that as in the purely classical case, the learning rate plays a large role, and care should be taken in estimating this hyper-parameter. In these experiments, we see that for the learning rates $\alpha = 0.005$ and $\alpha = 0.0005$, the results from $n$-shot doubly  and exact expectation value SGD do not differ significantly. As in the VQE setting, the $n$-shot DSGD algorithms require more \emph{epochs} to converge, but as we have seen in the previous section, as each epoch is significantly cheaper for small $n$ DSGD, it is again clear that single shot DSGD should definitely be used to rapidly find approximate solutions, before possibly increasing $n$, in conjunction with learning rate decay, to increase the quality of the solution. In this case the enhanced efficiency of small $n$ estimators is amplified by the size of the data-set, as each epoch consists of a single optimization step per data-set instance, and thus the measurement efficiency improvements per optimization step should be multiplied by the size of the data-set.

\begin{figure}
\begin{center}
      \includegraphics[width=\columnwidth]{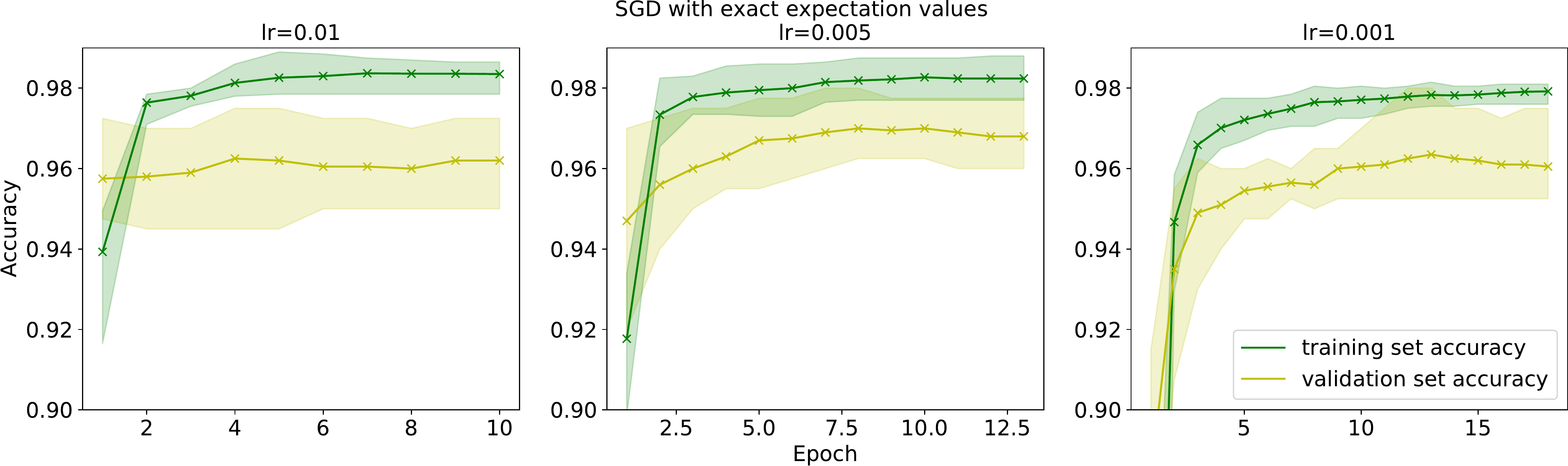}
      \includegraphics[width=\columnwidth]{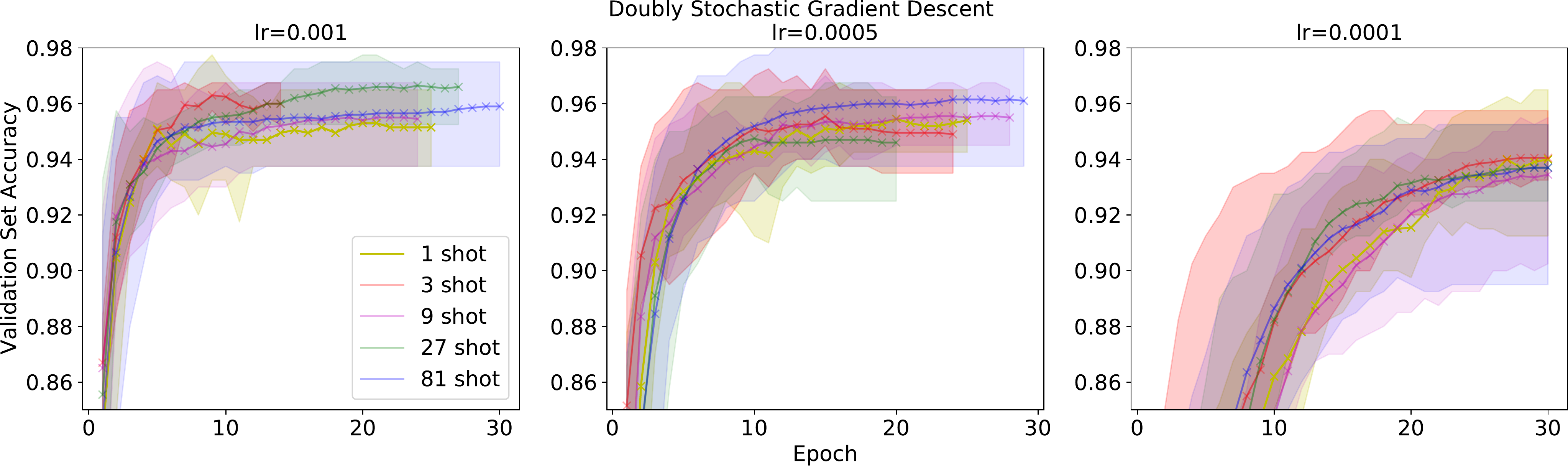}
\end{center}
  \caption{Results of downsampled MNIST binary ($3$ vs $6$) classification experiments. The results from $n$-shot doubly stochastic gradient descent (i.e. using the function \textproc{PartialDerivativeEstimator\_4} from Algorithm \ref{alg:MSE_loss functions}) are shown in the lower row, and  the results from exact expectation value stochastic gradient descent are shown in the upper row, both with batch-size = 1. All results were obtained using a fixed learning rate, displayed in the title of the respective subplot. All experiments were repeated 8 times, and the mean, max and min in each step is displayed. A particular experiment was decided to have converged if no improvement in the validation set accuracy was achieved in the previous 5 epochs, and for each shot number and learning rate pair the curve is plotted up until the epoch number which was reached by the slowest experiment to converge.}\label{fig:classifier_results}
\end{figure}

\section{Discussion and Conclusion}\label{s:conclusion}

In light of the above analysis and results, it is possible to draw a variety of observations and conclusions. Firstly, in the context of hybrid quantum-classical optimization, in which both loss functions and their gradients can be expressed as functions of expectation values, \emph{exact} gradient descent is not possible. As such, all optimization schemes which in practice require the estimation of expectation values from a finite number of measurement results should be considered within the framework of \emph{stochastic} gradient descent. In this work we have formalized this notion for VQE, QAOA and MSE classification problems, from which it is clear that previous approaches based on approximating expectation values with $n$ measurement outcomes can be understood as instances of well defined stochastic gradient descent type algorithms, which are valid for all $n$, and in particular, even for $n=1$. It is therefore not a-priori necessary to perform large numbers of measurements when implementing these algorithms in practice. Additionally, as is done classically in mini-batch stochastic gradient descent, whenever the loss function consists of a linear combination of terms which can be estimated in an unbiased way it is possible to obtain a new ``doubly stochastic" optimizer by sampling over terms of the linear combination. In the context of hybrid quantum-classical optimization such linear combinations appear from a variety of sources, such as sums over local Hamiltonian terms, sums over parameter shift terms, or sums over data-set instances. Exploiting this insight has allowed us to define multiple doubly stochastic gradient descent algorithms, for practically relevant settings such as VQE, QAOA and quantum classifiers.  Moreover, while convergence guarantees for realistic non-convex landscapes remain an open question, for simplified landscapes convergence guarantees are available for all of the algorithms discussed, which provide a minimum measure of confidence in their construction.  

From these convergence guarantees it is clear that even in restricted settings the accuracy of the solution to which a given stochastic gradient descent algorithm will converge is dependent on the variance of the estimator for the gradient. It is natural to expect that both small $n$ implementations of the algorithms we have introduced, as well as algorithms which involve samples over (possibly multiple) linear combinations, may not achieve solutions which are as fine-tuned as those obtained by large $n$ algorithms. From the numerical experiments we have performed in Section \ref{s:results} we see that this is indeed the case, however there are two important things to note. Firstly, adaptive heuristic learning rate schemes can successfully compensate for the variance of gradient estimators, and significantly improve the quality of solutions, even for very high variance estimators. Secondly, we see that on benchmark practical settings both small $n$ stochastic gradient descent algorithms, and algorithms that sample over linear combinations of terms, can converge orders of magnitude faster, with respect to the total number of measurements required. As a result, even though the solutions converged to are not optimal, it is well advised to utilize small $n$ SGD algorithms -- possibly with linear combination sampling -- as initializers which are able to rapidly find approximate solutions, using very few measurements relative to the number required by large $n$ algorithms to reach the same point. Once convergence has been achieved for small $n$ with sampling, optimization can be continued with a combination of learning rate decay and slowly increasing values of $n$, in order to fine tune the solutions. 

Given the results presented here, there remain a variety of directions to explore. From an analytical perspective, it is naturally of interest to construct unbiased estimators for loss functions such as the 
\emph{log-likelihood} and the \emph{cross entropy}. However, as previously discussed, despite the absence of rigorous constructions for \emph{unbiased} estimators for non-polynomial loss functions, one can  use both efficient expectation value estimation and sampling over linear combinations to constructed biased estimators for these loss functions, whose performance as heuristic strategies may be worth investigating further, and for which one still may be able to design optimization algorithms with rigorous convergence guarantees \cite{biased}. Additionally, It is also desirable to obtain convergence guarantees for settings in which the Polyak-Lojasiewicz inequality is not satisfied, although this is expected to require significantly new techniques, which are sought after by the classical machine learning community. It should also be noted that the results and experiments presented here have all neglected the effects of noise, and understanding how realistic noise influences these results is of natural importance for developing optimization methods for use in conjunction with existing devices. One possible scenario is that realistic device noise will lead to estimators which remain unbiased, but with a larger variance. In this case, the strategies suggested here remain valid, however one may ultimately require smaller learning rates and larger shot numbers for fine-tuning solutions. If however realistic device noise introduces a bias into the estimators, then as mentioned earlier, one could still apply the strategies and optimization algorithms developed here as heuristics, and it may also be possible to design more complicated algorithms which still admit convergence guarantees \cite{biased}. In order to understand the role of noise more fully, it would therefore be of interest both to study the performance of these algorithms on existing devices and to integrate realistic noise models into the analysis. In fact, along precisely these lines, very recently the authors of Ref. \cite{gentini2019noise} have shown that counter to intuition noise may in fact be beneficial for stochastic gradient descent optimization, and it is hoped that this work stimulates further future research in these directions. Finally, it would be interesting to combine the algorithms presented here with additional techniques for reducing the number of circuit executions and forward passes. Very recent work has in fact introduced a variety of heuristic hyper-parameter adaptation schemes for measurement shot adaptation \cite{kubler2019adaptive}, and it will be of interest to investigate the extension of these techniques to the additional algorithms we have presented here. It is also worthwhile to further exploit structure such as sparsity in the gradients estimated.

On a higher level, we hope that the rigorous framework provided in this work contributes in a significant fashion to the growing body of analytical work on variational quantum-classical hybrid algorithms and algorithms with applications in quantum-enhanced machine learning, complementing a -- to date -- still largely empirical field of research. At the same time, we hope that our work provides further perspectives to find practical applications of near-term quantum devices, let this be \emph{near-term quantum circuits} 
\cite{GoogleQubits,IBMQubits} or instances of \emph{programmable quantum simulators} \cite{InnsbruckQAOA,ProbingQuantumSimulator,Monroe}, beyond showing conceptually interesting 
quantum advantages or ``supremacy'' \cite{bremner_average-case_2016, BosonSampling, neill_blueprint_2017,NewSupremacy,GWD16, arute2019quantum}. Optimized schemes of the kind developed here may help in the search of finding such applications.

\begin{acknowledgments}
FW acknowledges funding from the DFG under Germany's Excellence Strategy – MATH+: The Berlin Mathematics Research Center, EXC-2046/1 – project ID: 390685689, and would like to thank Zacharias V. Fisches for insightful discussions on techniques in deep learning and numerical methods. RS acknowledges the support of the Alexander von Humboldt foundation, and would like to acknowledge the Quantum Excellence in Diversity (QuEDiver) initiative for facilitating a research visit of MS to the FU Berlin.
JE has been supported by the BMWi (PlanQK), the 
ERC (TAQ), the  Templeton  Foundation,  and  the  
DFG  (EI 519/14-1,  EI 519/15-1,  CRC 183) and MATH+.   
This work has also received funding from the European Union’s Horizon 2020 research and innovation programme under grant agreement No.~817482 (PASQuanS).
\end{acknowledgments}	

\bibliographystyle{abbrvunsrtnat}
\bibliography{literature.bib}

\clearpage
\appendix
\section{Proof of Theorem \ref{thm:lips_cont1}\label{app:lipschitz_proof}}

For completeness, in this section we provide a proof of Theorem \ref{thm:lips_cont1} via a sequence of Lemmas. We begin with the following standard result for univariate functions, which follows directly from the mean value theorem, and whose proof can be found in~\cite{Lipshitz}:

\begin{lemma}[Lipschitz continuity of univariate functions on closed domains]\label{lem:lips_cont_univariate}
Given some function $f:\mathbb{R}\rightarrow \mathbb{R}$, if $f$ is continuously differentiable, then for any closed interval $[a,b] \subset \mathbb{R}$, the function $f:[a,b]\rightarrow \mathbb{R}$ is $L$-Lipschitz, with $L = \sup_{x \in [a,b]} |f'(x)|$.
\end{lemma}
\noindent In order to continue, we introduce the following definition:
\begin{definition}\label{def:lc_wrt_j}
Given some function $f:[a,b]^M \rightarrow \mathbb{R}$, we say that $f$ is Lipschitz continuous with respect to its $j$'th argument, if for all $\vec{y} \in [a,b]^{M-1}$ the function $g_{j,\vec{y}}:[a,b] \rightarrow \mathbb{R}$ is Lipschitz continuous, where
\begin{equation}
    g_{j,\vec{y}}(x) := f(y_1,\ldots,y_{j-1},x,y_j,\ldots y_{M-1}).
\end{equation}
\end{definition}
\noindent Note that if $f$ is Lipschitz continuous with respect to its $j$'th argument, this implies that for all $x_1,x_2 \in [a,b]$ and for all $y \in [a,b]^{M-1}$ there exists some positive constant $L_{j,\vec{y}}$ such that
\begin{equation}
    |g_{j,\vec{y}}(x_2) - g_{j,\vec{y}}(x_2)| \leq L_{j,\vec{y}} |x_2 - x_1|.
\end{equation}
In light of this, we define $L_j := \sup_{\vec{y}} L_{j,\vec{y}}$, and it follows that for all $x_1,x_2 \in [a,b]$ and for all $y \in [a,b]^{M-1}$ 
\begin{equation}
|g_{j,\vec{y}}(x_2) - g_{j,\vec{y}}(x_2)| \leq L_j |x_2 - x_1|.
\end{equation}
Given, this we are then able to state the following Lemma, providing a sufficient condition for Lipschitz continuity of multivariate functions:
\begin{lemma}\label{lem:app_inter}
Given some $f:[a,b]^M \rightarrow \mathbb{R}$, if $f$ is Lipschitz continuous with respect to all of its arguments, then $f$ is $L$-Lipschitz continuous, with $L = \sqrt{M}(\max_{j}L_j)$.
\end{lemma}
\begin{proof}
For all $\vec{x},\vec{y} \in [a,b]^M$ we have that

\begin{align}
    f(\vec{x}) - f(\vec{y}) &= f(x_1,\ldots,x_M) - f(y_1,x_2,\ldots,x_M)  + \ldots\nonumber\\
    & \quad+ f(y_1,\ldots y_{j-1},x_j,\ldots,x_M) - f(y_1,\ldots y_j,x_{j+1},\ldots x_M) + \ldots \nonumber\\
    &\quad + f(y_1,\ldots y_{M-1},x_M) - f(y_1,\ldots y_M).
\end{align}
By the triangle inequality, this then gives us
\begin{align}
    |f(\vec{x}) - f(\vec{y})| &\leq |f(x_1,\ldots,x_M) - f(y_1,x_2,\ldots,x_M)|  + \ldots\nonumber\\
    & \quad+ |f(y_1,\ldots y_{j-1},x_j,\ldots,x_M) - f(y_1,\ldots y_j,x_{j+1},\ldots x_M)| + \ldots \nonumber\\
    &\quad + |f(y_1,\ldots y_{M-1},x_M) - f(y_1,\ldots y_M)|.
\end{align}
Furthermore, under the assumption that $f$ is Lipschitz continuous with respect to all of its arguments, we see that
\begin{align}
    |f(\vec{x}) - f(\vec{y})| &\leq L_1|x_1 - y_1| + \ldots + L_M|x_M - y_M| \\
    & \leq (\max_{j}L_j) (\sum_{j = 1}^M|x_j - y_j|)\\
    & = (\max_{j}L_j) ||\vec{x} - \vec{y}||_1\\
    & \leq \sqrt{M}(\max_{j}L_j)||\vec{x} - \vec{y}||_2,
\end{align}
where the last line follows from the fact that for $\vec{x} \in \mathbb{R}^M$ one has that $||\vec{x}||_1 \leq \sqrt{M}||\vec{x}||_2$.
\end{proof}\newpage
\noindent Putting together the previous statements we then obtain the following Lemma.
\begin{lemma}\label{final_lem}
Given some function $f:\mathbb{R}^M \rightarrow \mathbb{R}$, if all partial derivatives of $f$ are continuous, then for any $a,b \in \mathbb{R}$ the function $f:[a,b]^M \rightarrow \mathbb{R}$ is $L$-Lipschitz continuous with

\begin{equation}
    L = \sqrt{M}\left[\max_{j \in \{1,\ldots,M\}}\sup_{\vec{x} \in [a,b]^M}\left|\frac{\partial f(\vec{x})}{\partial x_j}\right|\right]
\end{equation}
\end{lemma}
\begin{proof}
For all $j$, and for all $y\in [a,b]^{M-1}$, one can apply Lemma \ref{lem:lips_cont_univariate} to the function $g_{j,\vec{y}}$, from which one obtains that $g_{j,\vec{y}}$ is $L_{j,\vec{y}}$-Lipschitz continuous with 
\begin{equation}
    L_{j,\vec{y}} = \sup_{x \in [a,b]}\left|\frac{\partial f(y_1,\ldots y_{j-1},x,y_j,\ldots y_{M-1})}{\partial x_j} \right|.
\end{equation}
The statement then follows from Lemma \ref{lem:app_inter} and the definition of $L_j$.
\end{proof}
\noindent Finally, given Lemma \ref{final_lem} the proof of Theorem \ref{thm:lips_cont1} follows straightforwardly:
\begin{proof}[Proof (Theorem \ref{thm:lips_cont1})]
For all $j$ we have that the partial derivatives
    \begin{equation}
        \frac{\partial}{\partial \theta_j}f(\vt)=\langle \vec{0}|\left( \frac{\partial U^\dagger(\vt)}{\partial \theta_j}\right)OU(\vt)|\vec{0}\rangle+\langle \vec{0}|U^\dagger(\theta)O\left( \frac{\partial U(\vt)}{\partial \theta_i}\right)|\vec{0}\rangle
    \end{equation}
exist and are continuous on $\mathbb{R}^M$, and therefore the statement of the theorem follows from Lemma \ref{final_lem}.
\end{proof}

\section{Parameterized circuit and optimization details}\label{app:parameterized_circuits}
Fig.~\ref{fig:vqe_circuit} below shows the parameterized quantum circuit $U(\vt)$ used for both the VQE and MSE quantum classifier experiments, as discussed in Section \ref{s:results}. As can be seen, the circuit is constructed from layers of Pauli rotations, interleaved with $CNOT$ ladders, and consists of 400 free parameters.

\begin{figure}[H]
\begin{center}
      \includegraphics[]{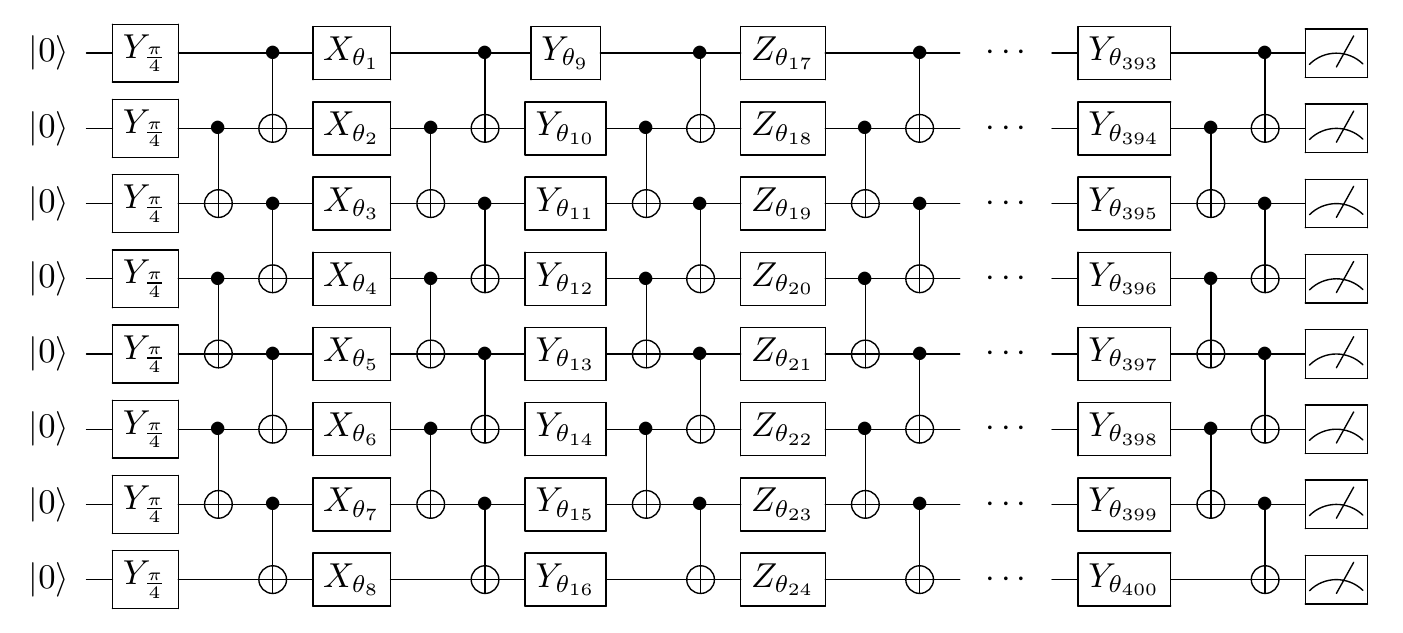}
\end{center}
  \caption{Parameterized quantum circuit used for the VQE and MSE classifier experiments described in Section \ref{s:results}. The notation $\sigma_\theta$ is used to represent the single qubit $e^{-i\theta\sigma}$ gate, for some Pauli operator $\sigma \in [X,Y,Z]$.}\label{fig:vqe_circuit}
\end{figure}
\newpage
\section{CO$_2$ Emission Table}

The table below summarizes the estimated carbon cost of this work, including both numerical simulations and air-travel for collaboration purposes. Estimations have been calculated using the examples of Scientific CO$_2$nduct \cite{conduct}, and are correct to the best of our knowledge.

\begin{table}[h]
\centering
\begin{tabular}[b]{l c}
\hline
\textbf{Numerical simulations} & \\
Total Kernel Hours [$\mathrm{h}$]& 14300\\
Thermal Design Power Per Kernel [$\mathrm{W}$]& 5.75\\
Total Energy Consumption Simulations [$\mathrm{kWh}$] & 85.1\\
Average Emission Of CO$_2$ In Germany [$\mathrm{kg/kWh}$]& 0.56\\
Total CO$_2$ Emission For Numerical Simulations [$\mathrm{kg}$] & 47.6\\
Were The Emissions Offset? & \textbf{Yes}\\
\hline
\textbf{Transport} & \\
\hline
Total CO$_2$ Emission For Transport [$\mathrm{kg}$] & 4804\\
Were The Emissions Offset? & \textbf{Yes}\\
\hline
Total CO$_2$ Emission [$\mathrm{kg}$] & 4847.5\\
\hline
\hline
\end{tabular}
\end{table}

\end{document}